%% file: main.tex
\documentclass[11pt]{article}
\usepackage[T1]{fontenc}
\usepackage[utf8]{inputenc}
\usepackage[margin=1in]{geometry}
\usepackage{amsmath,amssymb}
\usepackage[hidelinks]{hyperref}
\usepackage{graphicx}
\usepackage{tikz}
\usetikzlibrary{arrows.meta,decorations.pathreplacing}
\usepackage[hidelinks]{hyperref}
\usepackage{amsmath,amssymb,amsthm}

\usepackage{tcolorbox}
\tcbuselibrary{breakable}

\newtcolorbox{problem}[1]{
  colback=black!2,
  colframe=black!45,
  colbacktitle=black!8,
  coltitle=black,
  title={#1},
  fonttitle=\bfseries,
  boxrule=0.5pt,
  arc=2pt,
  left=10pt,
  right=10pt,
  top=8pt,
  bottom=8pt,
  before skip=12pt,
  after skip=12pt,
  breakable
}

\newtheorem{theorem}{Theorem}[section]
\newtheorem{lemma}[theorem]{Lemma}
\newtheorem{proposition}[theorem]{Proposition}
\newtheorem{corollary}[theorem]{Corollary}
\theoremstyle{definition}
\newtheorem{definition}[theorem]{Definition}
\theoremstyle{remark}

\newcommand{\PP}{\mathrm{P}}

\title{The Two-Dimensional Majority Rule is P-Complete}
\author{%
  Pedro Montealegre \qquad Martín Ríos-Wilson\\[0.5em]
  \small Facultad de Ingeniería y Ciencias\\
  \small Universidad Adolfo Ibáñez\\[0.3em]
  \small \href{mailto:p.montealegre@uai.cl}{\texttt{p.montealegre@uai.cl}}
  \quad \href{mailto:martin.rios@uai.cl}{\texttt{martin.rios@uai.cl}}%
}
\date{}

\begin{document}
\maketitle

\begin{abstract}
We prove that prediction for the synchronous two-dimensional majority
rule is $\mathrm{P}$-complete under logspace many-one reductions,
resolving a problem open for almost three decades. In 1997, Moore
established $\mathrm{P}$-completeness in dimension three and higher
and conjectured that the two-dimensional case admits an efficient
parallel algorithm. We consider an $n\times n$ torus on which each
cell follows the majority of its four nearest neighbors and retains
its current state in a tie. Given an explicitly specified initial
configuration and a time $T$, prediction asks
whether a designated cell is in state $+1$ at time $T$.

The central challenge is to make independent information streams
cross in the plane under a homogeneous, monotone, diffusive local
rule. We overcome this obstacle through a temporal encoding of
Boolean values: both values generate activity, but are distinguished
by signal arrival times. This encoding yields a crossover that
preserves both values and composes with wires, duplication, and
AND and OR gates to simulate arbitrary monotone Boolean circuits.
Thus a local rule that favors agreement can nevertheless transport,
combine, and cross independent information in two dimensions.

We also prove $\mathrm{P}$-completeness for deciding whether a
designated cell ever reaches $+1$, without a prescribed time horizon.
The prediction result extends to every uniform symmetric signed
majority rule on the same neighborhood, including the minority rule.
\end{abstract}

\section{Introduction}

How much computation can hide inside a majority rule?
A majority cellular automaton consists of a lattice of cells, each
carrying one of two states. At each time step, all cells update
simultaneously according to the majority state in their neighborhood.
The rule is local, monotone, and strongly biased toward agreement:
neighboring regions tend to reinforce one another rather than preserve
independent pieces of information. Majority dynamics also arise
naturally in statistical physics, through their relation with spin
evolution in the Ising model at zero temperature~\cite{moore1997majority}.

A natural way to measure the computational power hidden in these
dynamics is through \emph{prediction}: given an initial configuration,
how difficult is it to determine the state of a designated cell after
a prescribed number of updates? If configurations can encode Boolean
circuits and the evolution evaluates those circuits, then prediction
can reproduce the computations performed by the circuits themselves.
The complexity of prediction therefore provides a way of quantifying
the computational power of the dynamics.

We study this question for the synchronous majority rule on the
two-dimensional square lattice with periodic boundary conditions.
Each cell has four nearest neighbors; it follows their majority and
retains its current state in a tie. Equivalently, each update takes
the majority among five values: the state of the cell itself and those
of its four neighbors.

The long-term behavior of this dynamics is well understood. The
theorem of Goles and Olivos~\cite{goles1980periodic} implies that every
trajectory on a finite torus eventually reaches a fixed point or a
cycle of length two, and the decreasing-energy argument
of~\cite{goles1985decreasing} bounds the transient by $O(n^2)$ steps
on an $n\times n$ torus. Consequently, prediction belongs to
$\mathrm{P}$ even when the time horizon is encoded in binary: after
simulating the transient, only the phase of the eventual period-two
orbit matters.

The question is therefore where prediction lies \emph{within}
polynomial time. Could it be solved in logarithmic space? Could it be
efficiently parallelized? Or can the evolution encode arbitrary
polynomial-time computation? These possibilities are separated by the
 inclusions
\[
    \mathrm{L}\subseteq\mathrm{NL}\subseteq\mathrm{NC}
    \subseteq\mathrm{P}.
\]
Here $\mathrm{NC}$ captures problems admitting polylogarithmic-time
parallel algorithms using polynomially many processors. A problem is
$\mathrm{P}$-complete if it belongs to $\mathrm{P}$ and every problem
in $\mathrm{P}$ reduces to it by a logspace many-one reduction. Thus
P-completeness places a problem at the hardest end of polynomial time:
a logarithmic-space algorithm would imply $\mathrm{L}=\mathrm{P}$,
while an efficient parallel algorithm would imply
$\mathrm{NC}=\mathrm{P}$. We refer to
\cite{greenlaw1995limits} for the relevant complexity classes and
reductions.

In 1997, Moore~\cite{moore1997majority} proved that prediction for
majority-vote cellular automata is P-complete in dimension three and
higher. His construction embeds monotone Boolean circuits into the
dynamics: wires transport Boolean values, local configurations
implement logic gates, and a designated cell records the circuit
output. Since monotone circuit evaluation is
P-complete~\cite{goldschlager1977circuits}, the resulting dynamics can
encode arbitrary polynomial-time computation.

The two-dimensional case presents a fundamentally different
difficulty. In three dimensions, independent information flows can be
kept physically separated: one wire can simply pass over another.
In the plane, any simulation of a non-planar circuit must make two
signals meet in the same region and nevertheless emerge with their
values intact. For majority dynamics this is particularly delicate.
The same local interaction that propagates a signal also favors
agreement between nearby cells, so two crossing signals naturally
interfere rather than behave as independent channels.

There is also a complexity-theoretic reason to expect this geometric
restriction to matter. While monotone circuit evaluation is
P-complete in general, its planar restriction belongs to
$\mathrm{NC}$~\cite{delcher1995planar}. Based on this distinction,
Moore conjectured that prediction for the two-dimensional majority
rule admits an efficient parallel algorithm.

Almost three decades later, this conjecture remains open for the
standard majority rule. In the meantime, a sequence of works has
attacked closely related versions of the problem, obtaining
P-completeness or efficient parallel algorithms after modifying the
interaction graph, the local rule, the neighborhood, the
irreversibility of the dynamics, or the updating scheme; see
Section~\ref{sec:related-work}. Yet the homogeneous, unsigned,
synchronous majority rule on the two-dimensional square lattice has
remained unresolved.

We settle Moore's conjecture negatively, unless
$\mathrm{P}=\mathrm{NC}$. Prediction for the synchronous majority rule
on the two-dimensional torus is $\mathrm{P}$-complete. The key is a
temporal encoding of Boolean information. A Boolean value is not
represented solely by whether a signal is present, but by
\emph{when} a designated cell becomes active. This temporal degree of
freedom lets two information flows interact locally and still emerge
from a crossing with their values preserved. Together with wires,
duplication, and AND and OR gates, this yields a simulation of
arbitrary monotone Boolean circuits.

The significance of the construction is therefore not that complex
behavior can persist for a long time. Rather, it shows that a
homogeneous, monotone, diffusive local rule can transport,
combine, and even cross independent information streams in the plane,
with enough fidelity to realize P-complete computation.

\subsection{Prediction problems}
\label{subsec:prediction-problems}

We now define the problems formally. Let
$ V_n=(\mathbb{Z}/n\mathbb{Z})^2$, $n\geq 3,$
and let $x\in\{-1,+1\}^{V_n}$ be a configuration. Writing $u\sim v$
when $u$ and $v$ are nearest neighbors in the torus, define the
synchronous majority map $F$ by
\[
F(x)_v=
\operatorname{sgn}\!\left(
    x_v+\sum_{u\sim v}x_u
\right).
\]
The sum contains five terms in $\{-1,+1\}$ and is therefore never
zero. Equivalently, a cell follows the majority of its four neighbors
and preserves its state when those neighbors are evenly split.

Our first problem asks for the state of a designated cell at a
specified time.

\begin{problem}{Majority Prediction}
\textbf{Input:}
A configuration $x\in\{-1,+1\}^{V_n}$ given explicitly,
a target cell $v\in V_n$, and a time $T\geq0$ encoded in binary.

\medskip
\textbf{Question:}
Is $F^T(x)_v=+1$?
\end{problem}

We also consider prediction without a prescribed time horizon. Here
the question is whether the designated cell ever reaches state $+1$.

\begin{problem}{Majority Stability}
\textbf{Input:}
A configuration $x\in\{-1,+1\}^{V_n}$ given explicitly
and a target cell $v\in V_n$.

\medskip
\textbf{Question:}
Does there exist $t\geq0$ such that $F^t(x)_v=+1$?
\end{problem}

For a target initially in state~$-1$, \textsc{Majority Stability}
asks whether the cell ever changes to~$+1$; if the target is initially
in state $+1$, the answer is immediate. As discussed above, the
period-two property and the $O(n^2)$ transient bound imply that both
problems belong to $\mathrm{P}$.

\subsection{Our results}

Our main result settles Moore's two-dimensional prediction problem.
All completeness statements are under logspace many-one reductions.

\begin{theorem}
\label{thm:majority-prediction}
\textsc{Majority Prediction} is $\mathrm{P}$-complete.
\end{theorem}

The proof is a reduction from the Monotone Circuit Value Problem.
Its central difficulty is geometric: an arbitrary circuit cannot be
embedded in the plane without crossings. We overcome this obstacle
using the temporal encoding described above, which allows us to build
a crossover together with the other gadgets required for circuit
simulation.

The construction also yields hardness when no reading time is
specified. We attach a temporal fuse to the output of the simulated
circuit. The fuse converts the distinction between an early and a
late signal into a persistent change of state.

\begin{corollary}
\label{cor:eventual-prediction}
\textsc{Majority Stability} is $\mathrm{P}$-complete.
\end{corollary}

Finally, the equivalence established in~\cite{goles2018signed}
transfers the prediction result to uniform symmetric signed majority
rules. These include the minority rule, which follows the minority
among the four neighbors and reverses the current state in a tie,
and the \emph{crazy-spin} rule, which follows the majority of the
neighbors but reverses its current state in a tie.

\begin{corollary}
\label{cor:signed-prediction}
Prediction for every uniform symmetric signed majority rule
in two dimensions is $\mathrm{P}$-complete. In particular,
this holds for the minority and crazy-spin rules.
\end{corollary}

\paragraph{Interactive simulator.}
As supplementary material, we provide an interactive browser-based simulator
of the majority dynamics and circuit gadgets developed in this paper. The
simulator allows the reader to inspect the evolution of the individual gadgets
and to experiment with the circuit constructions used in the proof. It is
intended as a visual aid complementing the formal arguments and is available at
\url{https://mrioswilson.codefloe.page/majority2d-circuit-simulator/}.

\subsection{Related work}
\label{sec:related-work}
This work is contextualized within a long tradition of results on
majority dynamics, ranging from classical theorems on the asymptotic
behavior of threshold networks, a probabilistic approach including random initial conditions and random graphs,
computational complexity results, and more recent approaches based on
distributed verification. Over several decades, these works have
progressively identified structural conditions under which majority
dynamics can be understood, predicted, or efficiently verified.
Within this broader tradition, the authors
have pursued a line of research on the computational complexity of
majority dynamics, considering freezing rules, planar networks, signed
interactions, alternative updating schemes, and different neighborhood
geometries. In this work, we consider two forms of the prediction problem. In the
time-bounded version, an initial configuration and a time horizon $T$
are given, and the task is to determine the state of an objective cell
after $T$ synchronous updates. We also consider a version without a
fixed time horizon, in which the question is whether an objective cell
will ever change from its initial state during the evolution. The
following results provide the main context for these problems.

\paragraph{Dynamical properties.}
The long-term behavior of majority systems was studied before the
complexity of predicting their transient evolution. For finite
symmetric threshold networks, Goles and
Olivos~\cite{goles1980periodic} established the classical period-two
property: every trajectory eventually reaches a fixed point or a
cycle of period two. Moran~\cite{moran1995periodtwo} extended this
line to infinite locally finite graphs, proving the period-two
property under bounded degree and a subexponential growth condition.
Ginosar and Holzman~\cite{ginosar2000majority} subsequently developed
this analysis for infinite graphs and established local period-two
behavior under related growth conditions. Their framework includes,
in particular, finite-dimensional grid graphs such as
$\mathbb{Z}^d$. These results impose strong restrictions on the
asymptotic behavior of majority dynamics, while still allowing
nontrivial transient evolutions whose prediction may be
computationally difficult.

More recent work has investigated stabilization times on specific
graph families. Benjamini, Zakharov, and
Zhukovskii~\cite{benjamini2026trees} determined the exact worst-case
stabilization time of majority dynamics on arbitrary finite trees.
They also analyzed perfect cubic trees with uniformly random initial
configurations, obtaining asymptotic bounds on their typical
stabilization time. These results further illustrate how properties
of the underlying graph constrain the temporal behavior of majority
dynamics.

\paragraph{Computational complexity.}
The computational complexity of majority dynamics was initiated by
Moore~\cite{moore1997majority}, who proved that exact prediction for
majority-vote cellular automata is P-complete in dimension three and
higher. His construction simulates monotone Boolean circuits, using
the additional spatial dimension to route signals, and left the
corresponding prediction problem on the two-dimensional square
lattice open.

A first family of tractable cases was obtained for irreversible
variants of majority dynamics. For freezing majority, where a cell in
state $1$ remains in that state forever, Goles, Montealegre, and
Todinca~\cite{goles2013bootstrap} established an NC algorithm for
eventual activation under strict majority on graphs of maximum degree
at most four, a class that includes the square lattice with von Neumann
neighborhood. The same work obtained P-completeness results beyond the
corresponding degree threshold. Goles et
al.~\cite{goles2017freezing} subsequently proved that eventual
activation for freezing non-strict majority on the square lattice
also belongs to NC, using a characterization in terms of triconnected
components. Their broader study of binary freezing totalistic
cellular automata~\cite{goles2020totalistic} identified both NC and
P-complete problems within this family, providing a more general
classification of the effect of the local threshold rule.

The dependence on the precise neighborhood geometry has been explored
further in two dimensions. Concha-Vega et
al.~\cite{conchavega2026lshaped} studied freezing majority with
L-shaped neighborhoods on the square lattice. They obtained NC
algorithms for L-shaped neighborhoods of size two and P-completeness
results for other neighborhoods in the same family.

A complementary direction considered settings in which majority
dynamics can simulate circuits despite geometric restrictions. Goles
and Montealegre~\cite{goles2015planar} established P-completeness for
prediction on planar majority networks. In particular, their
construction shows that restricting the interaction graph to be
planar still permits the propagation and crossing mechanisms needed
for circuit simulation. Goles et al.~\cite{goles2018signed} later
studied signed majority cellular automata on the two-dimensional
grid, obtaining P-complete prediction and Turing universality for
uniform asymmetric and non-uniform symmetric rules.

The updating scheme constitutes another source of computational
complexity. Goles and Montealegre~\cite{goles2014updating} studied
threshold automata networks under synchronous, sequential, and
block-sequential updating schemes and obtained different complexity
bounds depending on the schedule. For majority networks under
block-sequential updates, Goles et al.~\cite{goles2016pspace} proved
PSPACE-completeness for deciding whether a designated vertex
eventually changes state. In the asynchronous setting, Goles and
Montealegre~\cite{goles2020asynchronous} considered the problem of
deciding whether a designated cell can change state with nonzero
probability under random asynchronous updates. They showed that this
problem belongs to NC on the two-dimensional square lattice with von
Neumann neighborhood, whereas it is NP-complete in dimension three
and higher.

Several of these results can be viewed as successive approaches to
Moore's two-dimensional question, obtained by restricting the
dynamics or modifying some aspect of the underlying model, such as
irreversibility, the interaction structure, or the updating scheme.
Complementarily, the present work returns to the homogeneous,
unsigned, synchronous majority rule on the two-dimensional square
lattice itself, with the same tie-preserving update at every cell,
and studies the worst-case complexity of exact prediction at a
specified time. We establish that this prediction problem is
P-complete.

\paragraph{Probabilistic setting.}
A parallel line of research studies majority dynamics under random
initial configurations or random underlying graphs. These works make
extensive use of probabilistic methods to investigate questions such
as information retention, consensus, coexistence, and typical
convergence times.

Tamuz and Tessler~\cite{tamuz2015retention} studied whether information
encoded in an initial collection of independently biased opinions is
preserved by majority dynamics. They obtained information-retention
results for broad families of bounded-degree finite and infinite
graphs and related the analysis to zero-temperature Ising dynamics.
Benjamini et al.~\cite{benjamini2016convergence} considered majority
dynamics on infinite unimodular transitive graphs and on random finite
graphs. For invariant random initial configurations on the former,
they proved almost-sure convergence or eventual period-two behavior.
They also obtained contrasting consensus and disagreement results for
different random graph models.

For two-dimensional cellular automata, G{\"a}rtner and
Zehmakan~\cite{gartner2021majority} studied majority and biased
majority dynamics on two-dimensional tori under independent random
initial configurations. They identified threshold regimes associated
with consensus and coexistence and analyzed the convergence time in
these regimes.

Random graph topologies have received considerable attention as well.
The work of Benjamini et al.~\cite{benjamini2016convergence} provided
results on consensus and disagreement for random graph models, while
more recently Goel and Sah~\cite{goel2026sparse} studied majority
dynamics on sparse $\mathbb{G}(N,p)$ graphs and obtained rapid
unanimity results, with high probability, over sparse parameter
regimes.

These probabilistic results address a different but complementary
aspect of majority dynamics. They describe typical behavior when the
initial configuration, the underlying graph, or both are sampled from
a probability distribution. Our setting instead concerns the
computational complexity of prediction in the worst case: the lattice
is fixed, the initial configuration is given explicitly, and the
complexity question asks how difficult it is to determine the state of
a designated cell after a prescribed number of updates.

\section{The gadgets}
\label{sec:constructions}

Our P-hardness reduction starts from the \emph{Monotone
Circuit Value Problem} (MCVP). Given an instance of MCVP, it constructs an initial configuration,
an output cell, and a reading time at which this cell is in state
$+1$ if and only if the circuit evaluates to $1$. We describe the elementary
pieces and the square patterns here. Section~\ref{sec:main-proof}
uses the grid-circuit construction of~\cite{goles2021asynchronous}
and explains the connections and reading times.

The construction has two levels. We first explain how information is
encoded and transported, using straight wires, crossings, and
right-angle turns. We then introduce AND and OR gates and a signal
duplicator. Each elementary piece is presented through its source-free
background and its evolution for every input assignment. Finally, we
assemble these pieces into uniform $k\times k$ blocks, with inputs on
the north and west sides and outputs on the south and east sides.

The common principle is a \emph{temporal encoding} of Boolean values.
A false value is not represented by the absence of a signal: both
values produce activity, but they are distinguished at a prescribed
reading time. These reading times are fixed by the construction,
independently of the input values.

\subsection{Wires and temporal encoding}
\label{subsec:wires}

We distinguish the physical state of a cell from the Boolean value
carried by a signal. Throughout Sections~\ref{sec:constructions}
and~\ref{sec:main-proof}, $x_v(t):=(F^t(x))_v\in\{-1,+1\}$
denotes the physical state of cell $v$ at time $t$, as in the figures;
a cell in state $+1$ is \emph{active}. Boolean values belong to
$\{0,1\}$ and are recovered by observing a designated cell at a
designated time. Uppercase $A,B$ denote local Boolean inputs in the
elementary figures; the uniform blocks below use $a$ for the north
input and $b$ for the west input. These are labels for Boolean values,
not additional physical states.

Each signal cell $v$ along an ordinary wire has a prescribed
\emph{reading time} $\tau(v)$. If the wire carries $b\in\{0,1\}$, then
\[
    x_v(\tau(v))=2b-1.
\]
Thus an active cell at its reading time represents \textsc{True},
whereas an inactive cell represents \textsc{False}. More precisely,
observing the cell in the same phase as its reading time gives
\begin{equation}
\label{eq:temporal-code}
x_v(t)=
\begin{cases}
-1, & t<\tau(v),\\
2b-1, & t=\tau(v),\\
+1, & t>\tau(v),
\end{cases}
\qquad t\geq0,\quad t\equiv\tau(v)\pmod2.
\end{equation}
In particular, whenever the displayed times are nonnegative,
\[
\begin{array}{c|ccc}
 & \tau(v)-2 & \tau(v) & \tau(v)+2\\ \hline
\textsc{True}  & -1 & +1 & +1\\
\textsc{False} & -1 & -1 & +1
\end{array}
\]
A true signal is active at the scheduled reading time; a false signal
becomes active at the next observation in the same phase, two updates
later. The comparison is made two steps apart because the activated
wire oscillates: a positive entry does not mean that the cell remains
positive at the intervening time. In particular, eventual activation
alone does not distinguish the Boolean values.

A straight wire consists of a one-cell-wide channel between two
checkerboard support strips, each two cells wide, with small solid
caps closing the ends of the displayed instance.
Figure~\ref{fig:wire-idle} shows its source-free background at $t=0$
and $t=1$. All cells, including those in the support, obey the same
synchronous majority rule; none are held fixed externally. Dark and
cyan cells distinguish background and signal roles, not different
update rules.
\input{gadgets/wire-idle.tex}

Along an oriented channel $v_0,\ldots,v_L$, the reading schedule
advances by one update per cell:
\begin{equation}
\label{eq:wire-reading-times}
    \tau(v_i)=\tau(v_0)+i,\qquad 0\leq i\leq L.
\end{equation}
Each cell therefore reproduces the same Boolean value at its own
reading time. The ports in Figure~\ref{fig:wire-dynamics} are eight
cells apart. The two rows show $A=1$ and $A=0$: at $T=8$, the output
is positive exactly when $A=1$, while at $T+2=10$ it is positive in
both cases.The wire implements $A\mapsto A$ without changing the
temporal code.

The finite source pattern shown in
Figure~\ref{fig:wire-dynamics} initiates the signal. The first channel cell is read at time $8$.

\input{gadgets/wire-dynamics.tex}

\subsection{Temporal crossover}
\label{subsec:crossover}

The crossover lets two signals pass through a common region while
preserving their respective values. Its horizontal output reproduces
the horizontal input $A$, and its vertical output reproduces the
vertical input $B$; thus it implements
\[
    (A,B)\longmapsto(A',B')=(A,B).
\]
Figure~\ref{fig:crossover-idle} shows the two phases of the
source-free background. The input routes meet at the central region;
the Boolean values must survive this interaction in their temporal
encoding, rather than in spatially disjoint channels.
\input{gadgets/crossover-idle.tex}

In the displayed instance, the left source initially leads by five
cells. The right arm has length thirteen, whereas the other arms
have length eight. This extension aligns the two output reading
times. Figure~\ref{fig:crossover-dynamics} shows the assignments
$(A,B)=(0,0),(1,0),(0,1),(1,1)$, in that order. Both outputs are read
at $T=12$, when they reproduce their corresponding inputs, including
in the two cases in which the input values differ. At $T+2=14$, both
output ports are positive in every case. The information is therefore
preserved at the scheduled reading time, not as a permanent
difference between the eventual states.
\input{gadgets/crossover-dynamics.tex}

\subsection{Right-angle turn}
\label{subsec:turn}

A right-angle turn changes the direction of propagation by $90^\circ$
while preserving the Boolean value. Figure~\ref{fig:turn-idle} shows
an enlarged instance in which each arm extends six cells from the
corner. The two frames display the source-free background at $t=0$
and $t=1$.
\input{gadgets/turn-idle.tex}

Figure~\ref{fig:turn-dynamics} shows both input values following the
channel around the corner. For $A=1$, the front reaches the corner
at $t=4$ and the output at its reading time $T=8$. For $A=0$, the
output is still negative at $T$ and becomes positive at $T+2=10$.
The turn thus implements $A\mapsto A$ with the same temporal code
as a straight wire. Rotated and reflected copies provide the
orientations needed to route signals inside the square blocks.
\input{gadgets/turn-dynamics.tex}

\subsection{AND gate}
\label{subsec:and}

The elementary AND gadget implements
\[
    (A,B)\longmapsto(A,A\land B).
\]
It preserves the first input at its horizontal output and computes
the conjunction at its vertical output.
Figure~\ref{fig:and-idle} shows the source-free background. In the
input configurations, the left source initially leads by three cells.
The right output arm has length ten, whereas the other arms have
length seven; this extension synchronizes the two outputs.
\input{gadgets/and-idle.tex}

Figure~\ref{fig:and-dynamics} shows the four assignments in the same
order as for the crossover. Both outputs are read at $T=10$. The
conjunction port is positive at this time only for $(A,B)=(1,1)$,
while the horizontal output reproduces $A$. At $T+2=12$, both ports
are positive for every assignment. The conjunction is therefore
encoded in the timing of the output signal, with a false result
represented by a late signal rather than an output that remains
inactive forever.
\input{gadgets/and-dynamics.tex}

\subsection{OR gate}
\label{subsec:or}

The elementary OR gadget implements
\[
    (A,B)\longmapsto A\lor B.
\]
Two horizontal input branches feed a vertical output branch. The
three channel arms have length seven, and the lower supporting
pattern is part of the autonomous background shown in
Figure~\ref{fig:or-idle}; it is not an additional variable input.
\input{gadgets/or-idle.tex}

The left source initially leads by two cells.
Figure~\ref{fig:or-dynamics} shows all four input assignments, with
the output read at $T=11$. The output is positive exactly when at
least one input is true. Only $(A,B)=(0,0)$ leaves it negative at
that time; its late activity is visible at $T+2=13$. As for AND,
the output uses the same temporal distinction as the input signals.
\input{gadgets/or-dynamics.tex}

\subsection{Signal duplicator}
\label{subsec:copy}

A circuit value may be needed at several subsequent gates. The
signal duplicator implements
\[
    A\longmapsto(A',A'')=(A,A),
\]
with one output continuing horizontally and another leaving
vertically. Figure~\ref{fig:copy-idle} shows the source-free
background. The solid lower support is part of this background,
not a second input signal.
\input{gadgets/copy-idle.tex}

The right output arm is extended by three cells so that both outputs
have the same reading time. Figure~\ref{fig:copy-dynamics} shows that
$A'$ and $A''$ both reproduce $A$ at $T=13$. At $T+2=15$, both ports
are positive for either input value. Repeated duplication makes a
signal available to several destinations; in particular, it lets us
supply two copies of the result of an AND or OR gate in the uniform
blocks below.
\input{gadgets/copy-dynamics.tex}

\subsection{Uniform square blocks}
\label{subsec:uniform-blocks}

The elementary pieces explain how signals propagate and interact.
We now package them into square blocks with a common spatial and
temporal interface. We follow the square-block presentation
of~\cite{goles2018signed} and the north/west-input,
south/east-output convention of~\cite{goles2021asynchronous}.
In the latter convention, both AND and OR deliver their result
on both output sides. This normalization concerns assemblies of
the elementary pieces, not a change to their individual functions.

\paragraph{Ports and logical interfaces.}
Fix the same odd side length $k=2\mu+1$ for every block, where
$\mu=(k-1)/2$ is the port offset. The symbol $m$ is reserved for the
number of circuit gates. Using
row/column coordinates, with rows increasing southward and columns
eastward, place the ports at
\[
    N=(0,\mu),\qquad W=(\mu,0),\qquad
    S=(k-1,\mu),\qquad E=(\mu,k-1).
\]
Inputs enter through $N$ and $W$, and outputs leave through $S$ and
$E$. Writing $a$ for the north input and $b$ for the west input, the
four block interfaces are
\[
\begin{array}{c|c}
\text{Block} & \text{Output values }(E,S)\\ \hline
\mathrm{AND}   & (a\land b,a\land b)\\
\mathrm{OR}    & (a\lor b,a\lor b)\\
\mathrm{CROSS} & (b,a)\\
\mathrm{DUP}_{\mathrm W} & (b,b)
\end{array}
\]
The west-input duplicator has its north side closed. Reflecting
across the north-west/south-east diagonal gives a north-input
variant $\mathrm{DUP}_{\mathrm N}$ with its west side closed and
outputs $(a,a)$. A closed
side is not a logical zero input: a zero, when required, must be
supplied as a late signal.

The figures pair the logical interface with the physical initial
configuration. All four realizations use
\[
    k=153,\qquad \Delta=298,
\]
where $\Delta$ is the common input-to-output delay, measured between
reading times. These are fixed constants, independent of the
simulated circuit, and are not claimed to be minimal. The left
panels show functional connections; the right panels show the
actual arrangement of cells inside the square.

\paragraph{AND and OR blocks.}
Figure~\ref{fig:block-and} connects the conjunction output of the
elementary AND to a duplicator, which sends the result to both $E$
and $S$. The elementary gate's additional passthrough output is
terminated inside the block. The distinction between the elementary
function $(A,B)\mapsto(A,A\land B)$ and the normalized function
$(a,b)\mapsto(a\land b,a\land b)$ is thus implemented by assembly.
\input{gadgets/block-and.tex}

Figure~\ref{fig:block-or} applies the same principle to OR. An
elementary OR computes $a\lor b$, and a duplicator supplies the two
output copies. Wires and turns connect these pieces to the prescribed
boundary ports.
\input{gadgets/block-or.tex}

\paragraph{Crossover and duplication blocks.}
Figure~\ref{fig:block-crossover} connects the north input to the
south output and the west input to the east output. Its central
interaction is the temporal crossover described above, so that the
two input values remain independent. Figure~\ref{fig:block-copy}
uses the west input and delivers two copies of its value, while
leaving the north side closed.
\input{gadgets/block-crossover.tex}
\input{gadgets/block-copy.tex}

\paragraph{Internal synchronization.}
\label{subsec:phase-converter}
\label{subsec:delay}
The compact elementary figures have different arm lengths and
reading times. Inside a uniform block, the routes are chosen so
that both outputs satisfy the temporal code at the common reading
time $\tau+\Delta$, when the inputs are read at $\tau$. In particular,
an output carrying zero is negative at $\tau+\Delta$ and positive at
$\tau+\Delta+2$, in the same physical phase.

Straight wires contribute one update of delay per channel cell, as
in~\eqref{eq:wire-reading-times}. Route-length adjustments align
the arrival times, and the reading phases must also match. A branch
of the duplicator that changes the reading phase can be used with
the other branch terminated, supplying an internal phase adjustment.
These adjustments and the matching of supporting patterns are
handled inside the blocks, rather than as additional elementary
logical operations.

\paragraph{Connecting blocks.}
At the circuit level, the uniform $k\times k$ blocks are arranged in a
rectangular grid. The east output of each block is adjacent to the west
input of its right neighbor, and its south output is adjacent to the
north input of the block below. Since the squares are disjoint, the two
port cells are adjacent rather than identified, so transferring a signal
from one block to the next takes one update.

All active ports use the same channel interface. Hence matching ports of
adjacent blocks form a continuation of the wire pattern. When a block
ignores one of its logical inputs, the corresponding incoming channel
ends at the fixed termination included in that block pattern. Exposed
outputs on the boundary of the complete construction are terminated in
the same way.

Blocks on the antidiagonal $i+j=\ell$ belong to the same circuit level.
Let $\tau^{\mathrm{in}}_\ell$ and $\tau^{\mathrm{out}}_\ell$ denote,
respectively, the common input and output reading times at level $\ell$.
Since every block has delay $\Delta$ and every inter-block connection
takes one update,
\[
    \tau^{\mathrm{out}}_\ell
      = \tau^{\mathrm{in}}_\ell+\Delta,
    \qquad
    \tau^{\mathrm{in}}_{\ell+1}
      = \tau^{\mathrm{out}}_\ell+1.
\]
Consequently, $\tau^{\mathrm{in}}_{\ell+1}
      = \tau^{\mathrm{in}}_\ell+(\Delta+1).$ 
The supporting checkerboard patterns are synchronized with this schedule.
Let $B_L^{(0)}$ and $B_L^{(1)}$ denote the two source-free phases of a
block with label $L$, expressed in local coordinates. A block in grid position $(i,j)$ uses the compatible phase
$B_L^{((i+j)\bmod 2)}.$ 

Indeed, translating a block by $(ik,jk)$ changes checkerboard parity by
$i+j$, because $k$ is odd. Since $\Delta+1$ is also odd, moving from one
grid level to the next changes spatial and temporal parity simultaneously.
The boundary blocks follow the same phase convention.

\section{Proof of the main result}
\label{sec:main-proof}

A \emph{logspace transducer} is a deterministic Turing machine with a
read-only input tape, a write-only output tape whose head only moves to
the right, and a work tape of which it uses $O(\log |w|)$ cells on input
$w$. A function computed by such a machine is \emph{logspace
computable}. A problem $A$ \emph{reduces} to $B$, written $A\le_L B$,
if there is a logspace-computable function $f$ such that
$w\in A$ if and only if $f(w)\in B$. A problem is
\emph{$\PP$-complete} if it belongs to $\PP$ and every problem in
$\PP$ reduces to it.

We use the following standard facts; see~\cite{greenlaw1995limits}.

\begin{lemma}\label{lem:logspace}
\begin{enumerate}
\item Logspace-computable functions are closed under composition. In
      particular, if $A$ is $\PP$-complete, $B\in\PP$, and $A\le_L B$,
      then $B$ is $\PP$-complete.
\item Let $f$ be a function such that $|f(w)|$ is polynomial in $|w|$,
      and both $w\mapsto |f(w)|$ and $(w,i)\mapsto f(w)_i$, the $i$-th
      symbol of $f(w)$, are logspace computable. Then $f$ is logspace
      computable.
\end{enumerate}
\end{lemma}

\subsection{Circuit value problems}

A \emph{Boolean circuit} is a sequence of gates $g_1,\dots,g_m$, where
each $g_i$ is either a constant in $\{0,1\}$, or $\neg g_p$, or
$g_p\wedge g_q$, or $g_p\vee g_q$ for some $p,q<i$. We write
$  \beta_i=\operatorname{val}(g_i)\in\{0,1\}$
for the inductively defined value of gate $g_i$; the output value of
the circuit is $\beta_m$. The \emph{Circuit Value Problem}
($\textsc{CVP}$) asks whether $\beta_m=1$.

A circuit is \emph{monotone} if it contains no negation gates, and
$\textsc{MCVP}$ denotes the restriction of $\textsc{CVP}$ to monotone
circuits. Input assignments are represented by constant gates. Throughout
this section, $m$ denotes the number of gates, $k$ the side length of a
physical block, and $n$ the side length of the torus.

\begin{theorem}[\cite{ladner1975circuit}]\label{thm:cvp}
$\textsc{CVP}$ is $\PP$-complete.
\end{theorem}

For our reduction, the monotone version suffices.

\begin{proposition}[\cite{goldschlager1977circuits}]\label{prop:mcvp}
$\textsc{CVP}\le_L\textsc{MCVP}$. Consequently, $\textsc{MCVP}$ is
$\PP$-complete, and to prove that a problem in $\PP$ is $\PP$-complete
it suffices to reduce $\textsc{MCVP}$ to it.
\end{proposition}

\subsection{The reduction scheme}
\label{subsec:reduction}

We first embed the monotone circuit into a purely logical grid circuit.
The construction follows the \emph{grid-embedded circuits}
of~\cite{goles2021asynchronous}, themselves inspired by Goldschlager's
layout of planar circuits~\cite{goldschlager1977circuits}. Information
flows only east and south, and every crossing is represented by a
temporal crossover.

The use of temporal crossings is important here. Monotone planar
circuits admit an $\mathsf{NC}$ evaluation
algorithm~\cite{delcher1995planar}; hence a constant-size monotone
Boolean crossover capable of planarizing arbitrary monotone circuits
would imply $\PP=\mathsf{NC}$. We use this observation only as
motivation: our construction crosses signals by separating them in
time.

The grid is organized by antidiagonals. Every east/south path from the
corner to a block $(i,j)$ traverses exactly $i+j$ blocks. Since all
physical blocks have the same delay $\Delta$, and matching ports of
adjacent blocks transfer a signal in one update, the reading time of a
block depends only on $i+j$.

\begin{figure}[t]
\centering
\includegraphics[width=\linewidth]{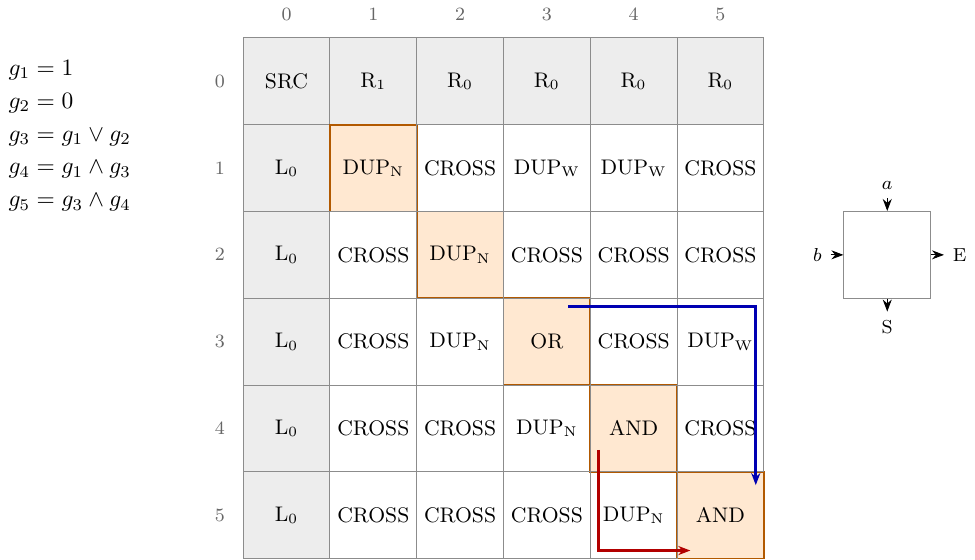}
\caption{Abstract grid embedding of the circuit on the left.
Gate $g_i$ is evaluated at $(i,i)$ and its value $\beta_i$ travels
east along row $i$ and south along column $i$. The inputs of $g_5$
turn at $\mathrm{DUP}_{\mathrm W}$ in $(3,5)$ and
$\mathrm{DUP}_{\mathrm N}$ in $(5,4)$. Other crossings are
$\mathrm{CROSS}$ blocks. The boundary clock starts at $\mathrm{SRC}$.
The reading schedule for block $(i,j)$, $i+j\ge1$, is
$T_{\mathrm{src}}+1+(i+j-1)(\Delta+1)$.}
\label{fig:grid}
\end{figure}

\begin{definition}[Grid circuit]\label{def:gridcircuit}
A \emph{grid circuit} of size parameter $m$ labels the $(m+1)^2$
positions $(i,j), 0\le i,j\le m,$
by blocks. Row indices increase southward and column indices eastward,
as in Section~\ref{subsec:uniform-blocks}.

An interior block reads a value $a\in\{0,1\}$ from the south output of
its north neighbor $(i-1,j)$ and a value $b\in\{0,1\}$ from the east
output of its west neighbor $(i,j-1)$. Its ordered output pair is
$(E,S)$. The logical functions of the blocks are listed in
Table~\ref{tab:blocks}.

\begin{table}[t]
\centering
\begin{tabular}{llccl}
\hline
Block & Name & East output & South output & Position\\
\hline
$\mathrm{SRC}$ & corner source & $1$ & $1$ & $(0,0)$\\
$\mathrm{R}_0$ & row boundary & $b$ & $0$ & row $0$, $j\ge1$\\
$\mathrm{R}_1$ & row boundary & $b$ & $1$ & row $0$, $j\ge1$\\
$\mathrm{L}_0$ & column boundary & $0$ & $a$ & column $0$, $i\ge1$\\
$\mathrm{CROSS}$ & crossover & $b$ & $a$ & interior\\
$\mathrm{DUP}_{\mathrm W}$ & west duplicator & $b$ & $b$ & interior\\
$\mathrm{DUP}_{\mathrm N}$ & north duplicator & $a$ & $a$ & interior\\
$\mathrm{AND}$ & conjunction & $a\wedge b$ & $a\wedge b$ & interior\\
$\mathrm{OR}$ & disjunction & $a\vee b$ & $a\vee b$ & interior\\
\hline
\end{tabular}
\caption{Logical grid interfaces. Inputs are ordered as $a$ from the
north and $b$ from the west, while outputs are ordered east, south.
The corner block has no inputs, $\mathrm{R}_c$ has only a west input,
and $\mathrm{L}_0$ has only a north input. A duplicator uses only one
logical input; the other incoming channel is terminated by the fixed
termination included in its physical block pattern. Boundary blocks
are used with incoming clock value $1$.}
\label{tab:blocks}
\end{table}

Row $0$ and column $0$ carry a logical clock of value $1$, emitted by
$\mathrm{SRC}$. A block $\mathrm{R}_c$ propagates this clock east and
supplies the value $c$ to the interior, whereas $\mathrm{L}_0$
propagates the clock south and supplies $0$ to the interior. The
circuit is evaluated along the antidiagonals $   i+j=0,1,\dots,2m.$

The corresponding majority patterns and their reading times are
described in Section~\ref{subsec:physical-layout}.
\end{definition}

\begin{lemma}\label{lem:grid}
There is a logspace-computable function mapping each monotone circuit
$g_1,\dots,g_m$ to a grid circuit of size parameter $m$ whose block
$(m,m)$ outputs $(\beta_m,\beta_m)$. Moreover, the label of each grid
position is logspace computable from the circuit and its coordinates.
\end{lemma}

\begin{proof}
Write every nonconstant gate as $ g_i=g_{p_i}\circ_i g_{q_i}, $
where $1\le p_i\le q_i<i$ and
$\circ_i\in\{\wedge,\vee\}$. Label the grid as follows:
\begin{itemize}
\item $(0,0)$ is $\mathrm{SRC}$. For $j\ge1$, $(0,j)$ is
      $\mathrm{R}_1$ when $g_j$ is the constant $1$, and
      $\mathrm{R}_0$ otherwise. For $i\ge1$, $(i,0)$ is
      $\mathrm{L}_0$.
\item $(i,i)$ is $\mathrm{AND}$ or $\mathrm{OR}$ according to
      $\circ_i$, and is $\mathrm{DUP}_{\mathrm N}$ when $g_i$ is
      constant.
\item For $1\le i<j$, $(i,j)$ is $\mathrm{DUP}_{\mathrm W}$ when
      $g_j$ is nonconstant and $i=p_j$, and is $\mathrm{CROSS}$
      otherwise.
\item For $1\le j<i$, $(i,j)$ is $\mathrm{DUP}_{\mathrm N}$ when
      $g_i$ is nonconstant and $j=q_i$, and is $\mathrm{CROSS}$
      otherwise.
\end{itemize}
The label of $(i,j)$ depends only on the coordinates and on the
descriptions of $g_i$ and $g_j$, and is therefore computable in
logarithmic space.

We prove by induction on $i$ that the east outputs of $(i,j)$ for
$j\ge i$, and the south outputs of $(r,i)$ for $r\ge i$, all equal
$\beta_i$.

When $g_i$ is constant, the boundary block $(0,i)$ supplies
$\beta_i$ down column $i$. All blocks above $(i,i)$ are crossovers, so
this value reaches the north duplicator at $(i,i)$, which produces two
copies of $\beta_i$.

Now let $g_i=g_{p_i}\circ_i g_{q_i}$. Above the diagonal, column $i$
consists of crossovers except at $(p_i,i)$, where a west duplicator
copies $\beta_{p_i}$ south. Similarly, to the left of the diagonal,
row $i$ consists of crossovers except at $(i,q_i)$, where a north
duplicator copies $\beta_{q_i}$ east. By the induction hypothesis,
these two values are available from the corresponding earlier gates.
The diagonal block therefore outputs $ \beta_{p_i}\circ_i\beta_{q_i}=\beta_i$
in both directions. This argument also covers the case $p_i=q_i$.

East of $(i,i)$, both crossovers and west duplicators preserve the
value arriving from the west. South of $(i,i)$, both crossovers and
north duplicators preserve the value arriving from the north. Hence
$\beta_i$ propagates along the remainder of row $i$ and column $i$,
completing the induction.
\end{proof}

\subsection{Grid circuit embedding}
\label{subsec:physical-layout}

We now replace every logical block by the corresponding uniform
$k\times k$ majority pattern from Section~\ref{sec:constructions}.
All active ports use the same channel interface. When two neighboring
blocks are placed in adjacent squares, matching ports form exactly a
continuation of the wire pattern: the channel continues across the
common side, and the support strips meet with the same checkerboard
phase. Consequently, a value read at an output port at time $t$ is
presented to the corresponding input port of the neighboring block at
time $t+1$.

The inward-facing ports of the boundary blocks use this same interface.
At an unused logical input of a duplicator, the incoming channel meets
the fixed termination included in the displayed block pattern. Thus
the complete grid is obtained by directly juxtaposing the finite block
patterns described in Section~\ref{sec:constructions}.

\paragraph{Boundary signals.}
The boundary interfaces of Definition~\ref{def:gridcircuit} are
realized using the duplicator and the temporal encoding
in~\eqref{eq:temporal-code}. Along the top boundary,
$\mathrm{R}_1$ is a west-input duplicator: one copy of the incoming
clock signal continues east, while the other is sent south into the
circuit.

The block $\mathrm{R}_0$ is obtained from $\mathrm{R}_1$ by modifying
only its south-going route: inside the block, this route makes a
two-cell detour before reaching the south port, increasing its channel
length by exactly two cells while leaving the east-going route
unchanged. Hence the south-going signal reaches the same port two
updates later than in $\mathrm{R}_1$, which is precisely the encoding
of Boolean value $0$ in~\eqref{eq:temporal-code}. The block
$\mathrm{L}_0$ is obtained by reflecting $\mathrm{R}_0$ across the
main diagonal.

At the upper-left corner, $\mathrm{SRC}$ is an
$\mathrm{R}_1$ block fed by the source described in
Section~\ref{subsec:wires}. Its input is read at time $8$, so its
two outputs have nominal reading time $ T_{\mathrm{src}}=8+\Delta=306.$
For the uniform blocks used here, $\Delta=298$, and hence $T_{\mathrm{src}}=306.$
Let
\[
    c_j=
    \begin{cases}
        1, & \text{if $g_j$ is the constant-$1$ gate},\\
        0, & \text{otherwise}.
    \end{cases}
\]
Each step along the boundary consists of one update across the common
port followed by the common block delay $\Delta$. Therefore
\[
    (0,j):\qquad
    (E,S)=(1,c_j)
    \quad\text{at time}\quad
    T_{\mathrm{src}}+j(\Delta+1),
\]
and
\[
    (i,0):\qquad
    (E,S)=(0,1)
    \quad\text{at time}\quad
    T_{\mathrm{src}}+i(\Delta+1).
\]
Here and below, an output value $0$ means that its signal is delayed by
two updates relative to the nominal reading time, according
to~\eqref{eq:temporal-code}.

\begin{lemma}[Composition along the grid]
\label{lem:block-composition}
For every grid position $(i,j)$, the Boolean outputs of its physical
block have nominal reading time
\[
    T_{\mathrm{src}}+(i+j)(\Delta+1).
\]
For every block with $i+j\ge1$, its available inputs have reading time
\[
    T_{\mathrm{src}}+1+(i+j-1)(\Delta+1).
\]
\end{lemma}

\begin{proof}
We induct on the level $\ell=i+j$.

At level $0$, the two outputs of $\mathrm{SRC}$ have reading time
$T_{\mathrm{src}}$ by construction.

Consider a block at level $\ell\ge1$. Every predecessor of this block
lies at level $\ell-1$, and therefore, by the induction hypothesis, its
relevant output is read at $ T_{\mathrm{src}}+(\ell-1)(\Delta+1).$
Matching ports transfer the value across the common block boundary in
one update. Hence every available input of the block is read at
$T_{\mathrm{src}}+1+(\ell-1)(\Delta+1).$

For an interior block the north and west inputs arrive simultaneously;
for a boundary block this formula applies to its unique incoming clock
signal.

Every uniform block has delay $\Delta$, so its outputs are read at
\begin{align*}
    T_{\mathrm{src}}
    +1+(\ell-1)(\Delta+1)+\Delta
    &=T_{\mathrm{src}}+\ell(\Delta+1).
\end{align*}
This proves both formulas.
\end{proof}

We now instantiate the construction with $  k=153$ and $\Delta=298, $
and therefore $ T_{\mathrm{src}}=306. $

Place block $(i,j)$ at translation $(ik,jk)$ in row/column
coordinates. Writing $\ell=i+j\ge1$, Lemma~\ref{lem:block-composition}
gives
\begin{align}
\tau_{\mathrm{in}}(i,j)
  &=T_{\mathrm{src}}+1+(\ell-1)(\Delta+1),
    \label{eq:grid-input-time}\\
\tau_{\mathrm{out}}(i,j)
  &=\tau_{\mathrm{in}}(i,j)+\Delta
   =T_{\mathrm{src}}+\ell(\Delta+1).
    \label{eq:grid-output-time}
\end{align}
Thus consecutive antidiagonals are separated by $\Delta+1=299$
updates.

Terminate the east-facing channels of the last column and the
south-facing channels of the last row with the fixed wire endings of
Section~\ref{sec:constructions}. Let $r_{\mathrm{ext}}$ be the maximum,
over the finite source and termination patterns used here, of the
distance by which such a pattern extends beyond the tiled
$(m+1)k\times(m+1)k$ square. Set $    c_{\mathrm{pad}}=2r_{\mathrm{ext}}+4$
and embed the complete assembly in the even-sided torus of side
\[
    n=
    2\left\lceil
        \frac{(m+1)k+c_{\mathrm{pad}}}{2}
      \right\rceil.
\]
All cells not belonging to the construction are initialized in
physical state $-1$.

Finally, translate the complete initial configuration so that the east
output port of block $(m,m)$ is at the origin. The translation is
applied to every assigned spin, including the checkerboard support
pattern; the phase is therefore preserved.

Since $(m,m)$ lies on level $2m$, the query time is
\begin{equation}
\label{eq:query-time}
    T=T_{\mathrm{src}}+2m(\Delta+1).
\end{equation}

\begin{lemma}[Uniform construction]\label{lem:layout}
The map that sends a monotone circuit $g_1,\dots,g_m$ to the torus side
$n$, the initial configuration $x$, and the query time $T$ defined
above is logspace computable. Moreover,
\[
    n=O(m)
    \qquad\text{and}\qquad
    T=O(m).
\]
\end{lemma}

\begin{proof}
The block library, source pattern, wire endings, and all geometric
offsets are fixed constants. Hence $n$, $T$, and the final translation
are computable using $O(\log m)$-bit arithmetic.

To compute the spin assigned to a given torus cell, first undo the
final translation modulo $n$. Determine whether the resulting cell
belongs to one of the fixed source or boundary-termination patterns.
Otherwise, when it lies inside the tiled square, division by the
constant $k$ yields its block coordinates $(i,j)$ and its local
coordinates inside that block. Lemma~\ref{lem:grid} gives the block
label in logarithmic space, while $i+j$ and the fixed global phase
determine its phase. The corresponding fixed block pattern then
determines the spin in $\{-1,+1\}$. Every remaining torus cell receives
state $-1$.

Thus both the output length and every one of the
$n^2=O(m^2)$ output spins are logspace computable. The result follows
from Lemma~\ref{lem:logspace}(2).
\end{proof}

\begin{corollary}\label{cor:layout-correct}
For the configuration $x$ and time $T$ of
Lemma~\ref{lem:layout}, $(F^T(x))_{(0,0)}=2\beta_m-1.$
In particular, the queried spin is $+1$ if and only if the input
monotone circuit evaluates to $1$.
\end{corollary}

\begin{proof}
Lemma~\ref{lem:grid} gives Boolean output $\beta_m$ at block $(m,m)$.
By Lemma~\ref{lem:block-composition}, its east output has nominal
reading time $    T=T_{\mathrm{src}}+2m(\Delta+1).$
This output port is translated to the origin. The temporal encoding
in~\eqref{eq:temporal-code} therefore gives
$(F^T(x))_{(0,0)}=2\beta_m-1.$
\end{proof}

Corollary~\ref{cor:layout-correct} and Lemma~\ref{lem:layout} give a
logspace reduction $ \textsc{MCVP}\le_L\textsc{Majority Prediction}.$
Since $\textsc{MCVP}$ is $\PP$-complete by
Proposition~\ref{prop:mcvp}, and \textsc{Majority Prediction} belongs
to $\PP$, this proves Theorem~\ref{thm:majority-prediction}.

\section{Extensions}\label{sec:extensions}

We now extend our results to the \textsc{Stability} problem and to variants of the
majority rule. First, a temporal fuse turns the reduction for
\textsc{Prediction} into a reduction for \textsc{Stability}, where no
time horizon is given. Second, a change of variables transfers the
$\PP$-completeness of \textsc{Prediction} to every uniform signed
majority rule on the von Neumann neighborhood whose interactions are
symmetric. This family contains,
among others, the minority rule and the crazy-spin rule, in which each
cell votes against its own state. Throughout this section we write
$e_1=(1,0)$ and $e_2=(0,1)$.

\subsection{Stability}\label{sec:stability}

The gadgets of Section~\ref{sec:constructions} transmit values that must
be read at prescribed times. A temporal fuse instead records an early
arrival as a permanent state. This is an additional output device; the
reduction for \textsc{Prediction} already reads its output at a
specified time and does not require permanent storage.

The fuse combines an input terminal $E$, a memory cell $M$, and an
autonomous clock that provides a finite capture window.
Figure~\ref{fig:fuse-idle} shows the clock eroding without an input.
Unlike the periodic idle backgrounds of
Section~\ref{sec:constructions}, the armed configuration changes even in
the absence of a signal.
\input{gadgets/fuse-idle.tex}

Figure~\ref{fig:fuse-dynamics} compares an early arrival, a late arrival,
and no arrival. An early arrival at $E$ at time $3$ sets $M$ to physical
state $+1$ permanently from time $4$. An arrival at time $5$, or no
arrival, leaves $M$ permanently in state $-1$. Thus the fuse converts a
distinction in arrival time into a persistent distinction at one cell.

\input{gadgets/fuse-dynamics.tex}

\begin{theorem}\label{thm:stability}
\textsc{Majority Stability} is $\PP$-complete.
\end{theorem}
\begin{proof}
Membership in $\PP$ is a consequence of a classical result \cite{goles1980periodic,goles1985decreasing} . For hardness we
reduce $\textsc{MCVP}$ to \textsc{Stability}, using the layout of
Lemma~\ref{lem:layout} with one modification: the terminator of the
east output of block $(m,m)$ is replaced by a channel of constant
length leading to the terminal $E$ of a fuse whose memory cell $M$
lies at the origin.

Let $T_E$ be the time at which a signal of value $1$ reaches $E$, so
that a signal of value $0$ reaches it at $T_E+2$. The fuse of
Figure~\ref{fig:fuse-idle} sets $M$ to $+1$ permanently when its input
arrives at time $3$, and leaves $M$ in state $-1$ forever when the input
arrives at time $5$ or never arrives. Its waiting window is bounded by
a clock, which is a path of cells of width one. We extend this path by
$2\ell=T_E-3$ cells, so that each cell of the extended part is adjacent
to exactly two cells of the path, except the new free end, and each
cell outside the path is adjacent to at most two cells of the path.

The extension must have even length, so $T_E$ must be odd. Since the
backgrounds have period two, the parity of $T_E$ is fixed by the route
from the east output of block $(m,m)$ to $E$. If $T_E$ has even parity, we insert the constant-size phase-adjustment
piece described in Section~\ref{subsec:phase-converter}. This piece
preserves the Boolean value and changes the reading phase by one
update. Hence, after this adjustment, $T_E$ is odd.

Under these conditions, every cell of the extended part other than the
free end, as well as the end of the original clock to which it is
attached, has three votes in state $+1$ and keeps its state; the free
end has only two and becomes $-1$; and every cell outside the path stays
$-1$. Hence the extension loses exactly one cell per step, while the
rest of the fuse repeats its two phases. At time $2\ell$ the extension
has disappeared, and the fuse coincides with its configuration at time
$0$ in Figure~\ref{fig:fuse-idle}. Consequently, $M$ changes its state
if and only if the input arrives at time $T_E$, that is, if and only if
$\beta_m=1$.

The path is extended straight down, along a column of the torus to the
right of the rectangle that contains no other cells in state $+1$. This
satisfies both adjacency conditions provided that $2\ell<n-c'$ for a
constant $c'$, so we enlarge $n$ to at least $2\ell+c'$. Since
$2\ell<T_E=T+O(1)$, $n$ remains polynomial in $m$. The basic fuse and the channel to $E$ have constant size, while the
added clock extension has length $2\ell=O(T)=O(m)$. $\ell$ is computed from $T$ with
$O(\log m)$-bit arithmetic, and a cell of the path is recognized by
comparing its coordinates with $\ell$. Hence, by the argument of
Lemma~\ref{lem:layout}, the reduction is logspace computable.
\end{proof}

\subsection{Symmetric signed majority}\label{sec:signed}

Signed majority cellular automata on the two-dimensional grid were
studied by Goles, Montealegre, Perrot and
Theyssier~\cite{goles2018signed}, who obtained P-complete prediction and
Turing universality for uniform asymmetric and non-uniform symmetric
rules. Here we consider the uniform symmetric case. A \emph{signed
majority rule} assigns a sign to each cell of the neighborhood and lets
every cell take the sign of the weighted sum of the spins in its closed
von Neumann neighborhood. The rule is uniform,
so the interaction between $v$ and $v+d$ has a weight $w(d)\in\{\pm1\}$
that does not depend on $v$. It is \emph{symmetric} if the interaction
between two cells has the same weight in both directions, that is,
$w(d)=w(-d)$. A symmetric rule is therefore determined by three signs:
the weight $a$ of the cell itself, the weight $h$ of its horizontal
neighbors and the weight $b$ of its vertical neighbors.

\begin{definition}\label{def:signed}
For $a,h,b\in\{-1,+1\}$, the \emph{signed rule} $G_{a,h,b}$ on
$\{-1,+1\}^{V_n}$ is
\[
 G_{a,h,b}(\sigma)_v=\operatorname{sign}\Bigl(a\,\sigma_v
 +h\bigl(\sigma_{v+e_1}+\sigma_{v-e_1}\bigr)
 +b\bigl(\sigma_{v+e_2}+\sigma_{v-e_2}\bigr)\Bigr).
\]
\end{definition}

The argument of the sign is a sum of five terms in $\{-1,+1\}$, so it is
never zero and $G_{a,h,b}$ is well defined. The majority rule is
$F=G_{+,+,+}$. Other members of the family are the \emph{minority rule}
$G_{-,-,-}=-F$; the \emph{antiferromagnetic rule} $G_{+,-,-}$; the
\emph{crazy-spin rule} $G_{-,+,+}$, in which a cell follows the majority
of its four neighbors when they are not tied and flips its state when
they split two against two; and the anisotropic rules with $h\ne b$,
which are ferromagnetic along one axis and antiferromagnetic along the
other.

The problems \textsc{Prediction} and \textsc{Stability} are defined
for $G_{a,h,b}$ as in Section~\ref{subsec:prediction-problems}, with $G_{a,h,b}$
in place of $F$; for \textsc{Prediction} the question is whether
$G_{a,h,b}^T(\sigma)_{(0,0)}=+1$. Both problems belong to $\PP$. Indeed, $G_{a,h,b}$ is a threshold network with a symmetric
integer weight matrix, and the period-two theorem of Goles and
Olivos~\cite{goles1980periodic} and the transient bound
of~\cite{goles1985decreasing} use only the symmetry of the weights, not
their signs. Since the sum of the absolute values of the weights is $5n^2$, they give the analogue
of the result of \cite{goles1980periodic,goles1985decreasing}  for $G_{a,h,b}.$

Hardness follows from \emph{switching}, the operation on signed graphs
that reverses the sign of every edge with exactly one endpoint in a
given set of vertices~\cite{zaslavsky1982signed}. On configurations, it
amounts to reversing the spins of the cells in that set. For
$\varepsilon\in\{-1,+1\}^{V_n}$ and $\sigma\in\{-1,+1\}^{V_n}$, let
$\varepsilon\sigma$ denote the configuration
$(\varepsilon\sigma)_v=\varepsilon_v\sigma_v$. For $p,q\in\{-1,+1\}$,
define
\[
 \varepsilon^{p,q}_{(i,j)}=p^{\,i}q^{\,j}\qquad\text{for }(i,j)\in V_n.
\]
This is well defined on the torus when $n$ is even, or when $p=q=+1$, in
which case $\varepsilon^{+,+}$ is constantly $+1$. In both cases,
$\varepsilon^{p,q}_{(0,0)}=+1$, $\varepsilon^{p,q}_{v\pm
e_1}=p\,\varepsilon^{p,q}_v$ and $\varepsilon^{p,q}_{v\pm
e_2}=q\,\varepsilon^{p,q}_v$.

\begin{lemma}\label{lem:switching}
Let $a,h,b\in\{-1,+1\}$ and $\varepsilon=\varepsilon^{ah,ab}$, and
assume that $\varepsilon$ is well defined. Then, for every
$\sigma\in\{-1,+1\}^{V_n}$ and every $t\ge0$,
\[
 G_{a,h,b}^t(\varepsilon\sigma)=a^t\,\varepsilon F^t(\sigma).
\]
\end{lemma}
\begin{proof}
We first prove the case $t=1$. Let $v\in V_n$. Since
$h\,\varepsilon_{v\pm e_1}=h\cdot ah\,\varepsilon_v=a\varepsilon_v$ and,
likewise, $b\,\varepsilon_{v\pm e_2}=a\varepsilon_v$, every term in the
definition of $G_{a,h,b}(\varepsilon\sigma)_v$ carries the factor
$a\varepsilon_v$:
\[
 G_{a,h,b}(\varepsilon\sigma)_v
 =\operatorname{sign}\Bigl(a\varepsilon_v\bigl(\sigma_v
 +\textstyle\sum_{u\sim v}\sigma_u\bigr)\Bigr)
 =a\,\varepsilon_v F(\sigma)_v.
\]
For the induction step, note that $F$ is odd, $F(-\sigma)=-F(\sigma)$,
since the sign function is odd. Assuming the statement for $t$,
\[
 G_{a,h,b}^{t+1}(\varepsilon\sigma)
 =G_{a,h,b}\bigl(\varepsilon\,a^tF^t(\sigma)\bigr)
 =a\,\varepsilon F\bigl(a^tF^t(\sigma)\bigr)
 =a^{t+1}\varepsilon F^{t+1}(\sigma).\qedhere
\]
\end{proof}

In words, after reversing the spins on the sublattice selected by
$\varepsilon$ and, when $a=-1$, at every odd time step, the dynamics of
any signed symmetric rule is exactly the majority dynamics.

\begin{theorem}\label{thm:signed}
For every $a,h,b\in\{-1,+1\}$, \textsc{Prediction} for $G_{a,h,b}$ is
$\PP$-complete. If $a=+1$, \textsc{Stability} for $G_{a,h,b}$ is
$\PP$-complete as well.
\end{theorem}
\begin{proof}
Membership was shown above. For hardness, we first observe that the
reductions of Lemma~\ref{lem:layout} and Theorem~\ref{thm:stability}
can be assumed to produce an even side $n$: if $n$ is odd, we add one row and one column of cells in
state $-1$ to the margin. Hence $\varepsilon=\varepsilon^{ah,ab}$ is
well defined on all instances, and $\varepsilon_{(0,0)}=+1$.

For \textsc{Prediction}, map an instance $(n,\sigma,T)$ of the majority
rule to $(n,\sigma',T)$ with $\sigma'=a^T\varepsilon\sigma$. By
Lemma~\ref{lem:switching} and the oddness of $F$,
\[
 G_{a,h,b}^T(\sigma')
 =a^T\varepsilon F^T(a^T\sigma)
 =a^{2T}\varepsilon F^T(\sigma)
 =\varepsilon F^T(\sigma),
\]
so $G_{a,h,b}^T(\sigma')_{(0,0)}=F^T(\sigma)_{(0,0)}$. For
\textsc{Stability} with $a=+1$, map $(n,\sigma)$ to
$(n,\varepsilon\sigma)$. By Lemma~\ref{lem:switching}, the state of the
origin at every time $t\ge0$ is then $F^t(\sigma)_{(0,0)}$.

In both cases, the bit of the new configuration at $(i,j)$ is the bit
of $\sigma$ at $(i,j)$, complemented according to the parities of $i$,
$j$ and, for \textsc{Prediction}, of $T$. These parities are read from
the binary representations, so the composed reduction remains logspace
computable by Lemma~\ref{lem:logspace}.
\end{proof}

For the minority rule, $ah=ab=+1$, so $\varepsilon$ is constantly $+1$,
no parity condition on $n$ is needed, and Lemma~\ref{lem:switching} reduces
to the identity $G_{-,-,-}^t=(-1)^tF^t$. For the antiferromagnetic and
crazy-spin rules, $\varepsilon$ is the checkerboard
$(-1)^{i+j}$.

\section{Discussion}
\label{sec:discussion}

We show that a homogeneous, diffusive local rule can preserve and cross independent streams of information in two dimensions. The key in this construction was to encode information in time. Even though the rule is monotone, this temporal encoding gets around the routing obstruction without changing the interaction graph and without using different local rules at different cells.

As a consequence, the prediction problem for the synchronous, tie-preserving majority rule on the von Neumann neighborhood is settled. However, several natural variants remain open.

\paragraph{Neighborhoods, tie-breaking, and updating schemes.}
Three modifications are particularly interesting. First, one can
replace the von Neumann neighborhood by the \emph{Moore neighborhood},
which consists of the eight surrounding cells, including the four
diagonal neighbors. The additional interactions change both the
possibilities for routing signals and the conditions under which
supporting patterns remain stable. Our construction therefore does
not immediately extend to this neighborhood.

One can also consider the \emph{Toom neighborhood}: this "L"-shaped neighborhood in which each cell has a north and east neighbor was studied in the freezing case by \cite{conchavega2026lshaped} and it was found to have an NC prediction problem. Nevertheless, to our best knowledge, the prediction problem for the strict majority rule remains open.

Second, one can replace the tie-preserving rule by a
\emph{biased majority rule}: whenever the neighboring votes are
evenly split, the cell adopts a fixed preferred state, either $+1$
or $-1$, regardless of its current state. The same preference is
used at every cell. The two choices of preferred state are related
by exchanging $+1$ and $-1$ throughout the configuration.

Third, one can replace synchronous updates by a \emph{sequential
updating scheme}. Here the input includes a permutation
$\pi=(v_1,\ldots,v_{n^2})$ of the cells. Each sweep updates the cells
one at a time in this order, using the most recently available
neighboring states, and the same permutation is repeated in every
sweep. Prediction then asks for the state of a designated cell
after a specified number of complete sweeps. In particular, the
schedule is given as part of the input; it is not chosen
nondeterministically during the evolution.

These modifications can be considered individually or in any
combination. Does prediction remain $\mathrm{P}$-complete, or do
some of these variants admit an efficient parallel algorithm?
We conjecture that the hardness persists throughout this family.

\paragraph{Conjecture.}
For each choice of  von Neumann,  Moore neighborhood,  tie-preserving
or uniformly biased majority rule, and synchronous or sequential
updating, the prediction problem on an explicitly specified
two-dimensional torus is $\mathrm{P}$-complete under logspace
many-one reductions. In the sequential case, the updating permutation is part of the input. Observe that this matters because the variant that asks whether there exists an update scheme making a given cell change its state has already been studied \cite{goles2014updating}.

The challenge is to determine how robust the information-processing
mechanisms developed here are under these changes. Enlarging the
neighborhood introduces additional interactions, changing the
tie-breaking rule modifies the supporting dynamics, and sequential
updates alter the timing on which the gadgets rely. A proof of the
conjecture would establish that the computational power exhibited
here is not confined to one particular combination of these choices.

\paragraph{Connections with sandpile prediction.}
The study of biased-majority dynamics also invites comparison
with prediction problems for the two-dimensional Abelian sandpile
model. In this model, each cell carries a nonnegative integer
number of grains. A cell containing at least four grains is
unstable and may \emph{topple}, sending one grain to each of its
four nearest neighbors and losing four grains itself. On a finite
square with dissipative boundary conditions, grains sent outside
the square are lost. Repeated legal topplings eventually produce
a stable configuration, independently of the order in which
unstable cells are selected~\cite{moore1999sandpiles}.

Prediction problems for sandpiles concern different aspects of
this evolution. One may ask about the final configuration after
stabilization, or whether a designated cell ever topples.
A common avalanche formulation starts from a stable
configuration, adds one grain at a specified cell, and asks
whether the resulting avalanche reaches a given
target~\cite{formenti2020sandpiles}. Under synchronous toppling,
one may instead ask whether the target topples at a prescribed
time. This timed formulation retains information that is not
captured by the eventual-toppling
question~\cite{conchavega2025timed}.

The complexity of sandpile prediction exhibits a striking
dependence on dimension. Moore and Nilsson established
$\mathrm{P}$-completeness in dimension three and higher and
membership in $\mathrm{NC}$ in dimension one, leaving the
two-dimensional case unresolved~\cite{moore1999sandpiles}.
Despite progress on related models and neighborhoods, the
complexity of the classical prediction problem for the standard
two-dimensional von Neumann sandpile remains
open~\cite{formenti2020sandpiles,conchavega2025timed}.

These questions provide a broader motivation for investigating
biased-majority prediction. Although the update mechanisms
differ, both settings invite questions about how local threshold
interactions can support the propagation and combination of
information. Signal crossing and synchronization provide
particularly natural points of comparison. Understanding which
parts of our construction can be adapted to biased tie-breaking
may also offer insights into the difficulties encountered in
sandpile models. More generally, exploring these parallels could
help clarify the scope and limitations of temporal encodings,
and distinguish mechanisms specific to majority dynamics from
those that may be useful in other lattice-based systems.

\section*{Acknowledgements}
The writing of this paper and the derivation of its results were assisted by AI agents. We used Claude Code (Anthropic) with the models Claude Sonnet 5.5 and Claude Opus 5.5, and Codex (OpenAI) with the models GPT-6 Astra and GPT-6.1 Sol. The authors reviewed all text and results and take full responsibility for them.

\bibliographystyle{plain}
\bibliography{biblio}

\end{document}

%% file: gadgets/wire-idle.tex
\begin{figure}[tbp]
\centering
\includegraphics[width=\linewidth]{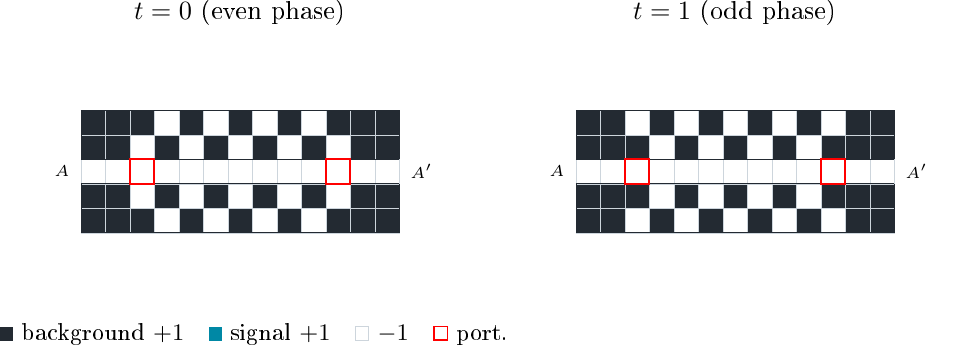}
\caption{Straight wire without input sources, at $t=0$ and $t=1$. The two frames show the phases of the autonomous idle background. All cells obey synchronous majority of self and the four nearest neighbors; no cells are externally fixed.}
\label{fig:wire-idle}
\end{figure}

%% file: gadgets/wire-dynamics.tex
\begin{figure}[tbp]
\centering
\includegraphics[width=\linewidth]{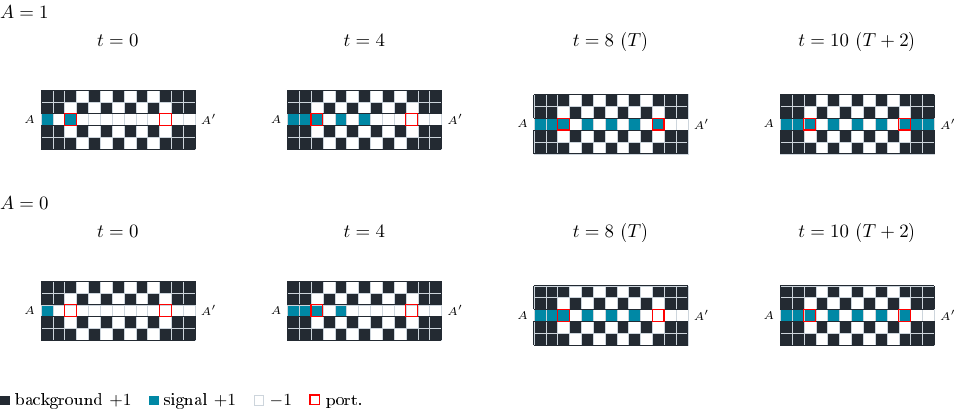}
\caption{Straight wire dynamics. The straight channel has two checkerboard support strips, each two cells wide, with solid end caps. The ports are eight cells apart; the output reproduces $A$ at $T=8$. Rows show $A=1$ and $A=0$, respectively. Outputs are sampled at $T$ and $T+2$. At $T$, positive and negative spins encode one and zero; the late signal is also positive at $T+2$.}
\label{fig:wire-dynamics}
\end{figure}

%% file: gadgets/crossover-idle.tex
\begin{figure}[tbp]
\centering
\includegraphics[width=\linewidth]{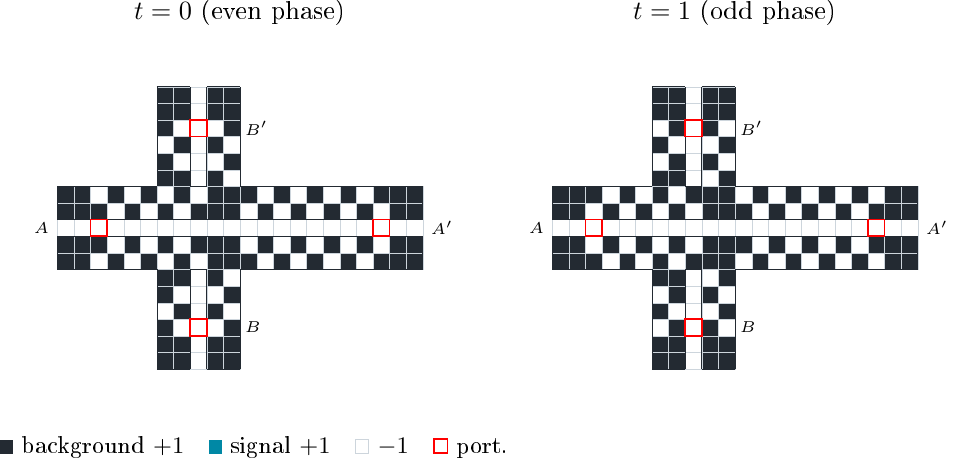}
\caption{Temporal crossover without input sources, at $t=0$ and $t=1$. The two frames show the phases of the autonomous idle background. All cells obey synchronous majority of self and the four nearest neighbors; no cells are externally fixed.}
\label{fig:crossover-idle}
\end{figure}

%% file: gadgets/crossover-dynamics.tex
\begin{figure}[tbp]
\centering
\includegraphics[width=\linewidth]{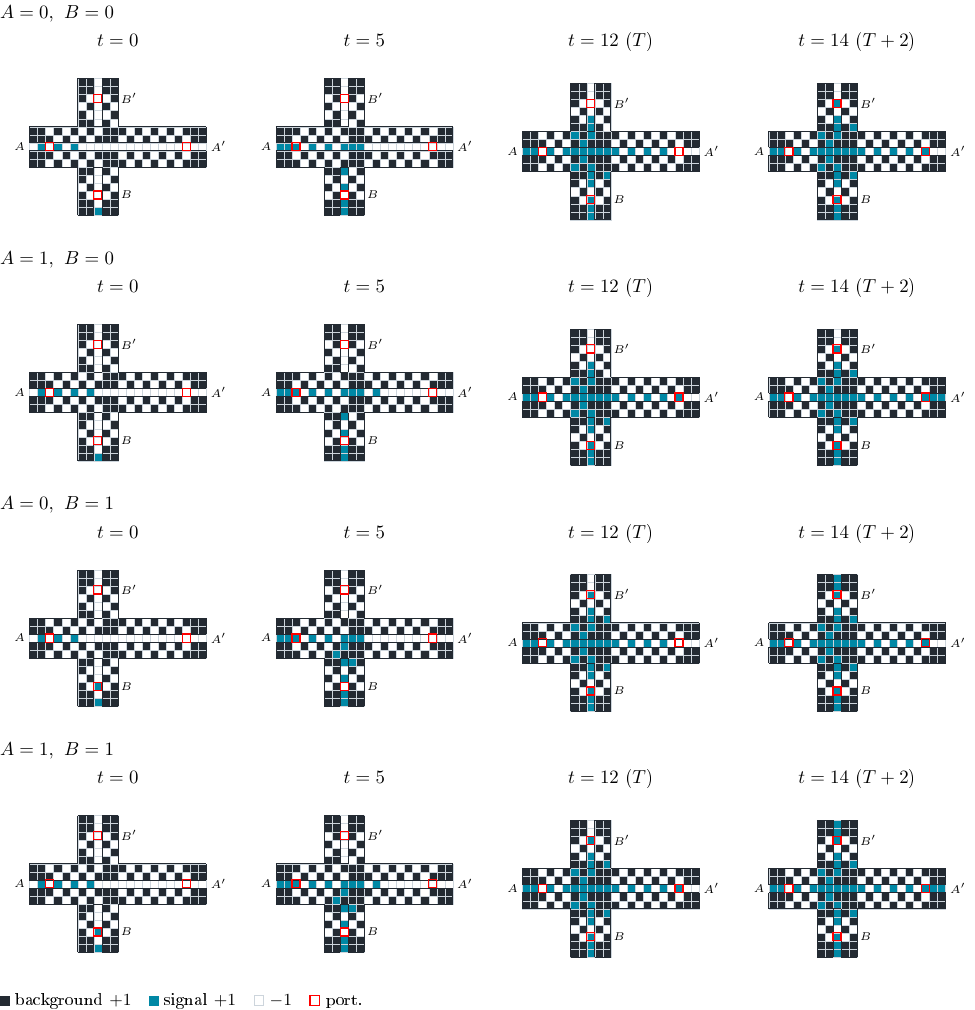}
\caption{Temporal crossover dynamics. The right arm has length thirteen; the other arms have length eight. The left source initially leads by five cells. Both outputs $(A\prime,B\prime)$ are read at the common time $T=12$. Rows are ordered as $(A,B)=(0,0),(1,0),(0,1),(1,1)$. Outputs are sampled at $T$ and $T+2$. At $T$, positive and negative spins encode one and zero; the late signal is also positive at $T+2$.}
\label{fig:crossover-dynamics}
\end{figure}

%% file: gadgets/turn-idle.tex
\begin{figure}[tbp]
\centering
\includegraphics[width=\linewidth]{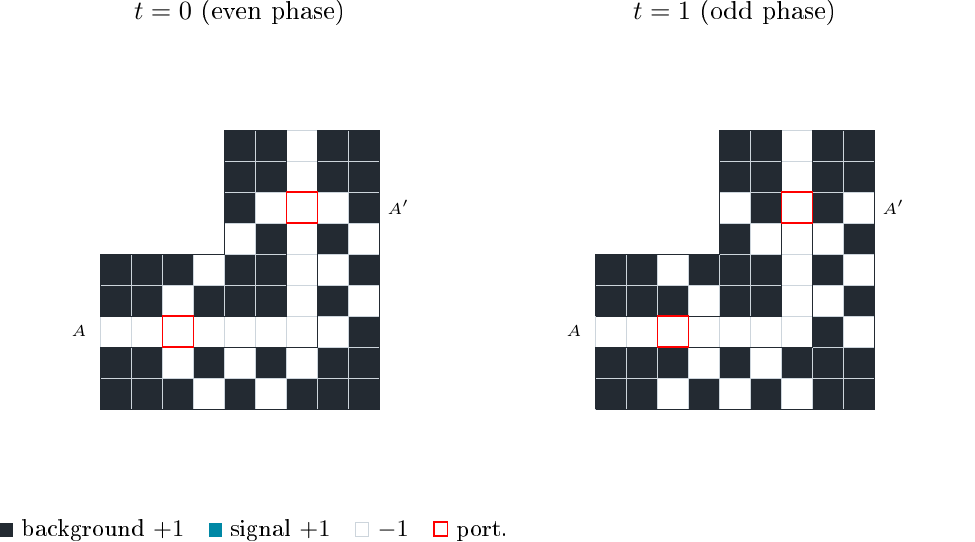}
\caption{90-degree turn without input sources, at $t=0$ and $t=1$. The two frames show the phases of the autonomous idle background. All cells obey synchronous majority of self and the four nearest neighbors; no cells are externally fixed.}
\label{fig:turn-idle}
\end{figure}

%% file: gadgets/turn-dynamics.tex
\begin{figure}[tbp]
\centering
\includegraphics[width=\linewidth]{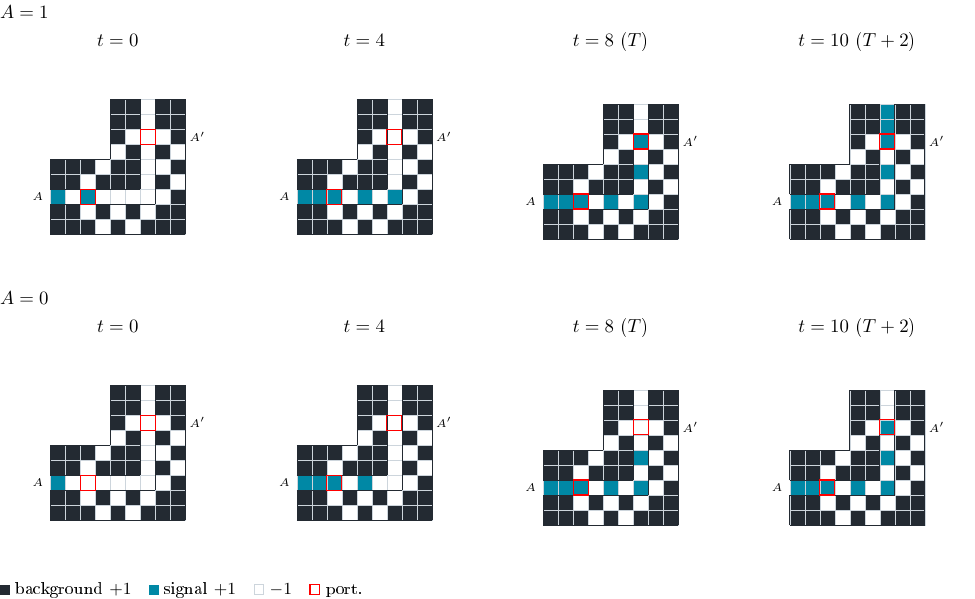}
\caption{90-degree turn dynamics. Each arm extends six cells from the corner, making the bend visible. For $A=1$, the front reaches the corner at $t=4$. The output is read at $T=8$. Rows show $A=1$ and $A=0$, respectively. Outputs are sampled at $T$ and $T+2$. At $T$, positive and negative spins encode one and zero; the late signal is also positive at $T+2$.}
\label{fig:turn-dynamics}
\end{figure}

%% file: gadgets/and-idle.tex
\begin{figure}[tbp]
\centering
\includegraphics[width=\linewidth]{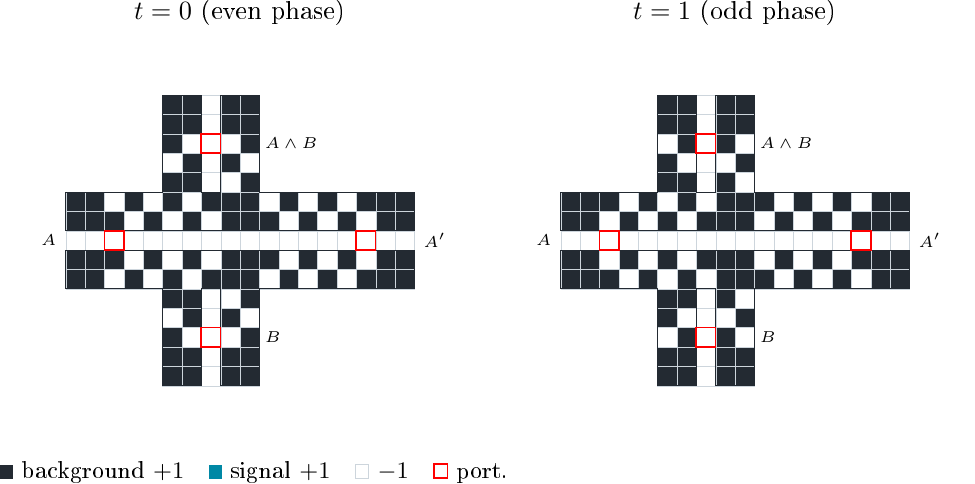}
\caption{AND gate without input sources, at $t=0$ and $t=1$. The two frames show the phases of the autonomous idle background. All cells obey synchronous majority of self and the four nearest neighbors; no cells are externally fixed.}
\label{fig:and-idle}
\end{figure}

%% file: gadgets/and-dynamics.tex
\begin{figure}[tbp]
\centering
\includegraphics[width=\linewidth]{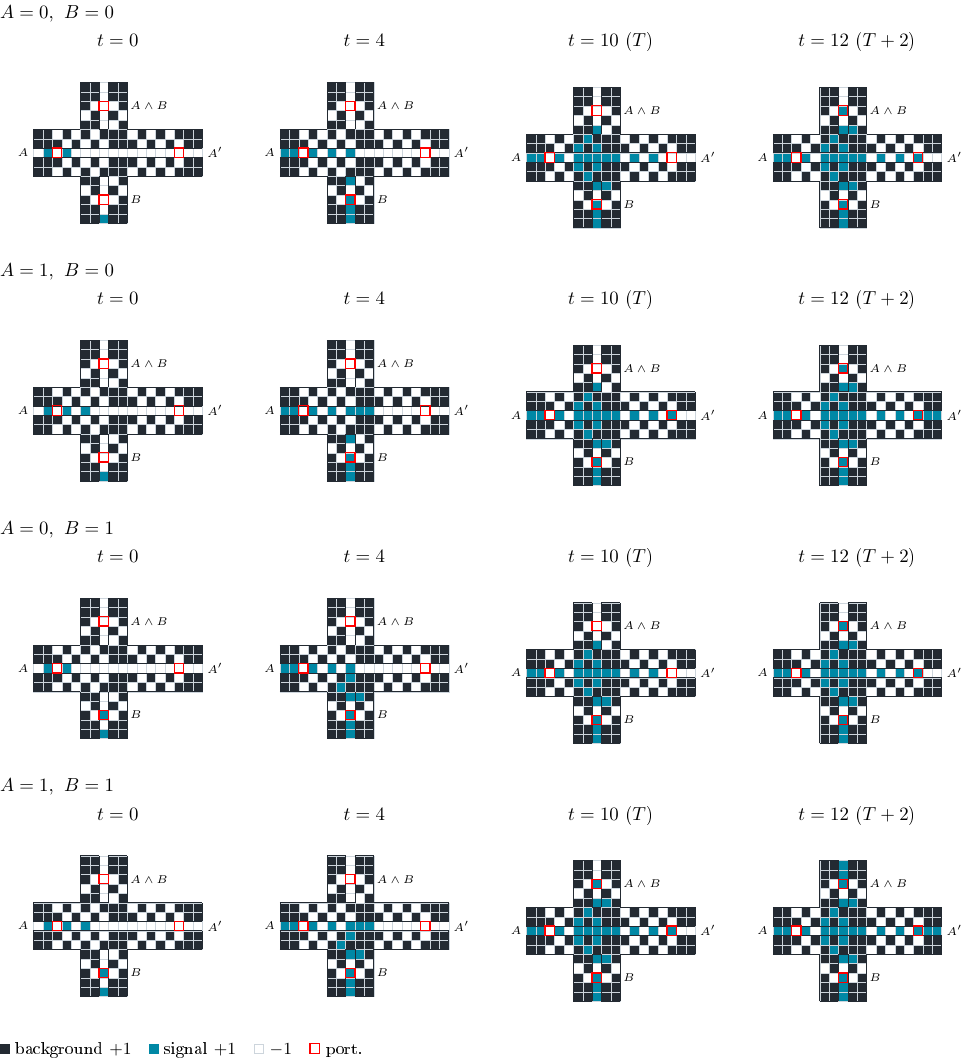}
\caption{AND gate dynamics. The right arm has length ten; the other arms have length seven. The left source initially leads by three cells. Both outputs $(A\prime,A\land B)$ are read at the common time $T=10$. Rows are ordered as $(A,B)=(0,0),(1,0),(0,1),(1,1)$. Outputs are sampled at $T$ and $T+2$. At $T$, positive and negative spins encode one and zero; the late signal is also positive at $T+2$.}
\label{fig:and-dynamics}
\end{figure}

%% file: gadgets/or-idle.tex
\begin{figure}[tbp]
\centering
\includegraphics[width=\linewidth]{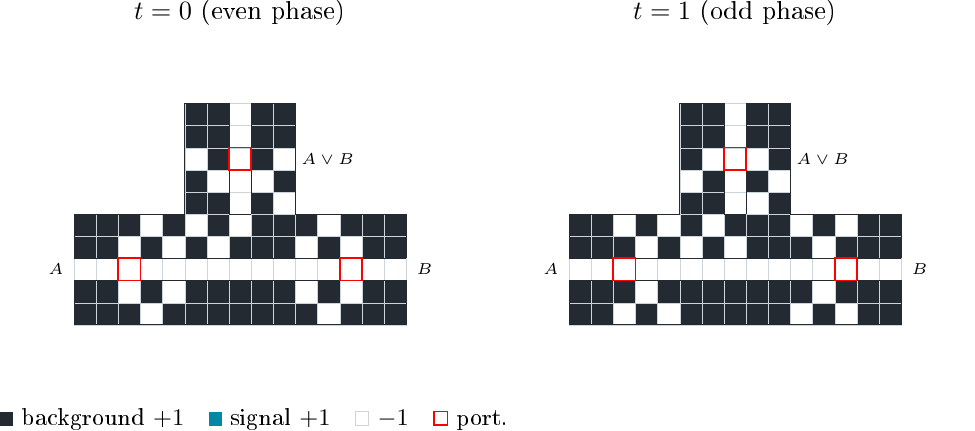}
\caption{OR gate without input sources, at $t=0$ and $t=1$. The two frames show the phases of the autonomous idle background. All cells obey synchronous majority of self and the four nearest neighbors; no cells are externally fixed.}
\label{fig:or-idle}
\end{figure}

%% file: gadgets/or-dynamics.tex
\begin{figure}[tbp]
\centering
\includegraphics[width=\linewidth]{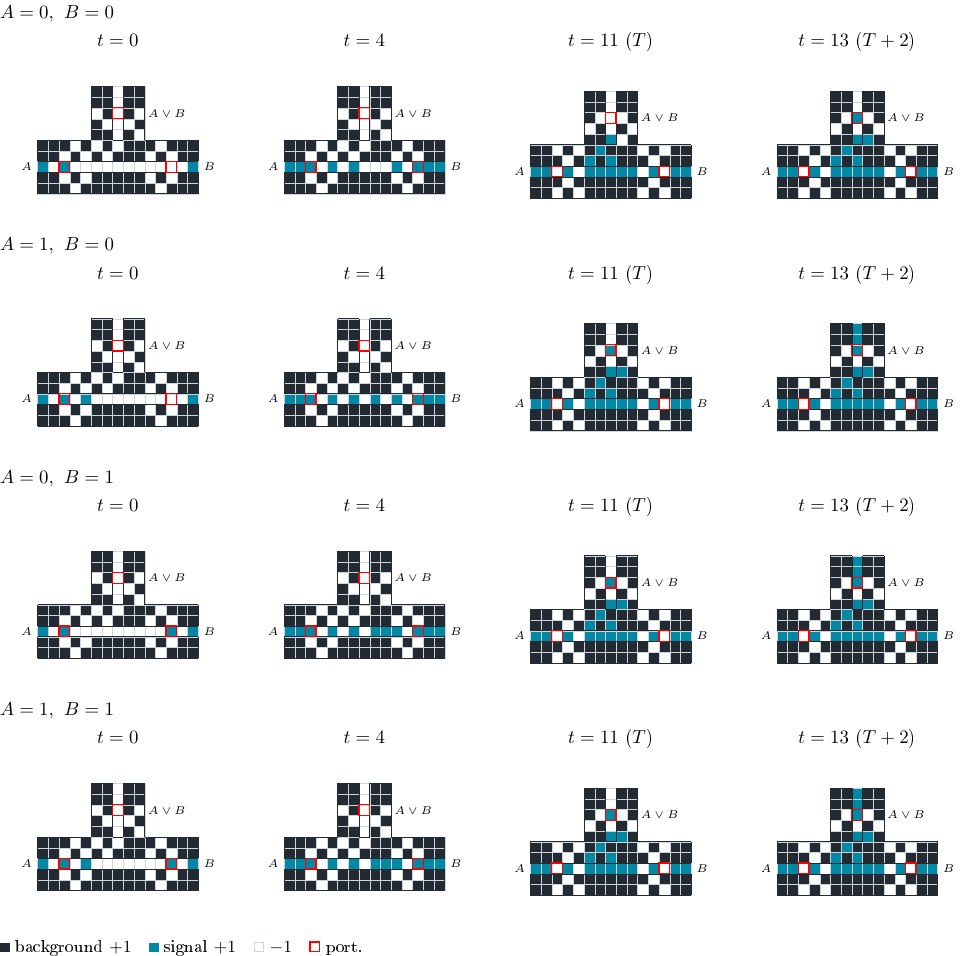}
\caption{OR gate dynamics. The three channel arms have length seven. The left source initially leads by two cells. The output $A\lor B$ is read at $T=11$. Rows are ordered as $(A,B)=(0,0),(1,0),(0,1),(1,1)$. Outputs are sampled at $T$ and $T+2$. At $T$, positive and negative spins encode one and zero; the late signal is also positive at $T+2$.}
\label{fig:or-dynamics}
\end{figure}

%% file: gadgets/copy-idle.tex
\begin{figure}[tbp]
\centering
\includegraphics[width=\linewidth]{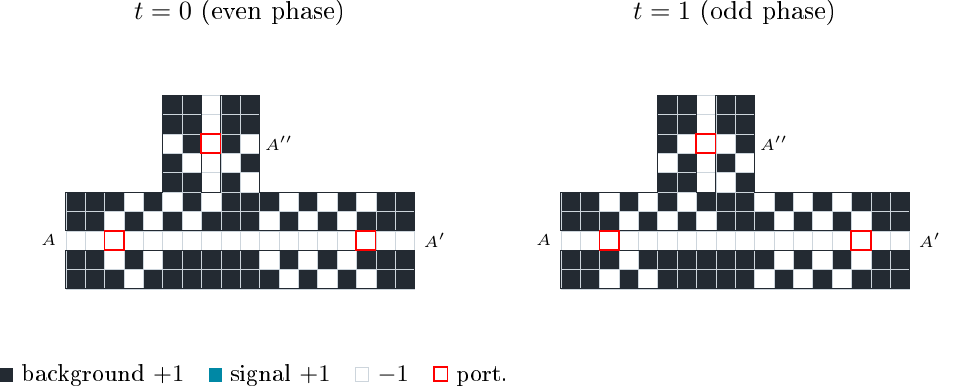}
\caption{Signal duplicator without input sources, at $t=0$ and $t=1$. The two frames show the phases of the autonomous idle background. All cells obey synchronous majority of self and the four nearest neighbors; no cells are externally fixed.}
\label{fig:copy-idle}
\end{figure}

%% file: gadgets/copy-dynamics.tex
\begin{figure}[tbp]
\centering
\includegraphics[width=\linewidth]{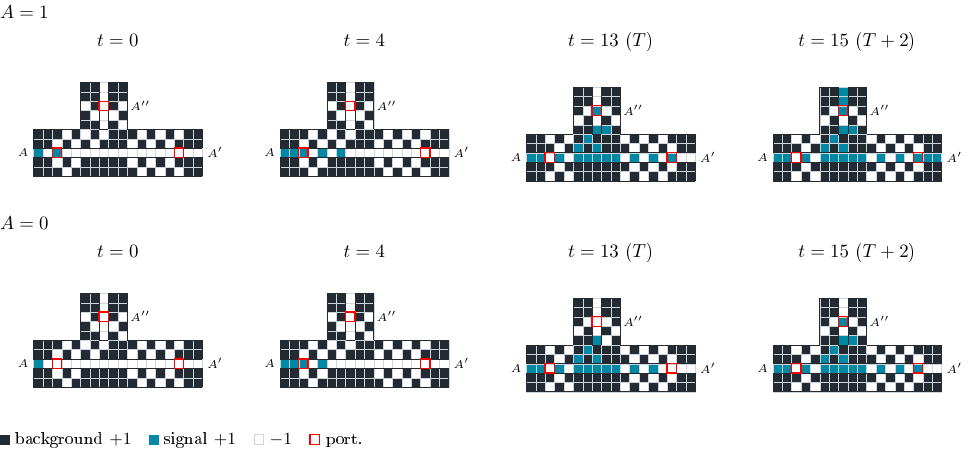}
\caption{Signal duplicator dynamics. The right arm is extended by three cells. Both outputs $(A\prime,A\prime\prime)$ reproduce $A$ at the common time $T=13$. Rows show $A=1$ and $A=0$, respectively. Outputs are sampled at $T$ and $T+2$. At $T$, positive and negative spins encode one and zero; the late signal is also positive at $T+2$.}
\label{fig:copy-dynamics}
\end{figure}

%% file: gadgets/block-and.tex
\begin{figure}[tbp]
\centering
\includegraphics[width=\linewidth]{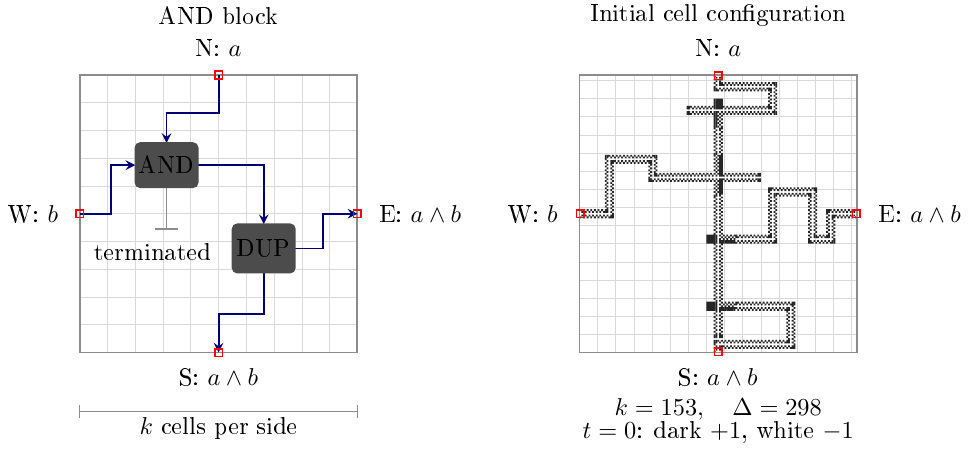}
\caption{Uniform AND block with $(E,S)=(a\land b,a\land b)$. The conjunction output of the elementary AND feeds a duplicator; the unused passthrough output terminates inside the block. Left: logical assembly. Right: the initial spins of the $153\times153$ realization; the input-to-output delay is $\Delta=298$ on both outputs. External sources and sink caps lie outside the square. Port markers are enlarged; the physical grid shows every tenth line.}
\label{fig:block-and}
\end{figure}

%% file: gadgets/block-or.tex
\begin{figure}[tbp]
\centering
\includegraphics[width=\linewidth]{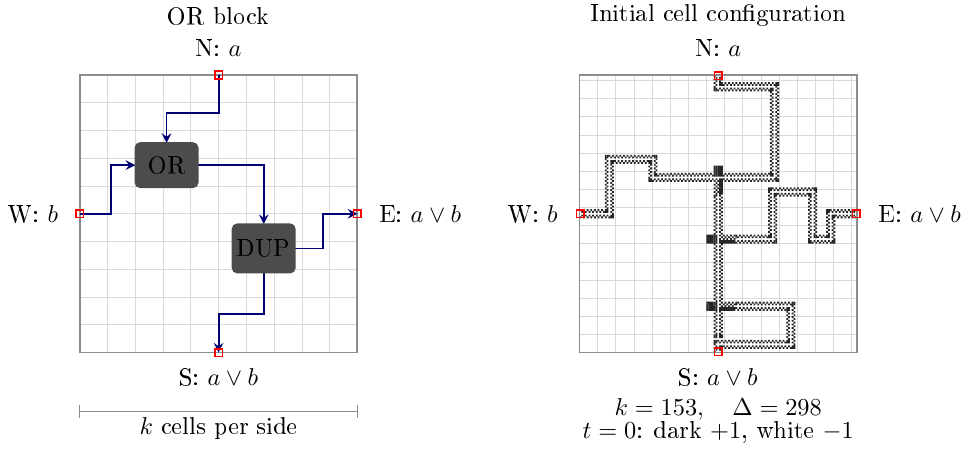}
\caption{Uniform OR block with $(E,S)=(a\lor b,a\lor b)$. The elementary OR feeds a duplicator to supply both output sides. Left: logical assembly. Right: the initial spins of the $153\times153$ realization; the input-to-output delay is $\Delta=298$ on both outputs. External sources and sink caps lie outside the square. Port markers are enlarged; the physical grid shows every tenth line.}
\label{fig:block-or}
\end{figure}

%% file: gadgets/block-crossover.tex
\begin{figure}[tbp]
\centering
\includegraphics[width=\linewidth]{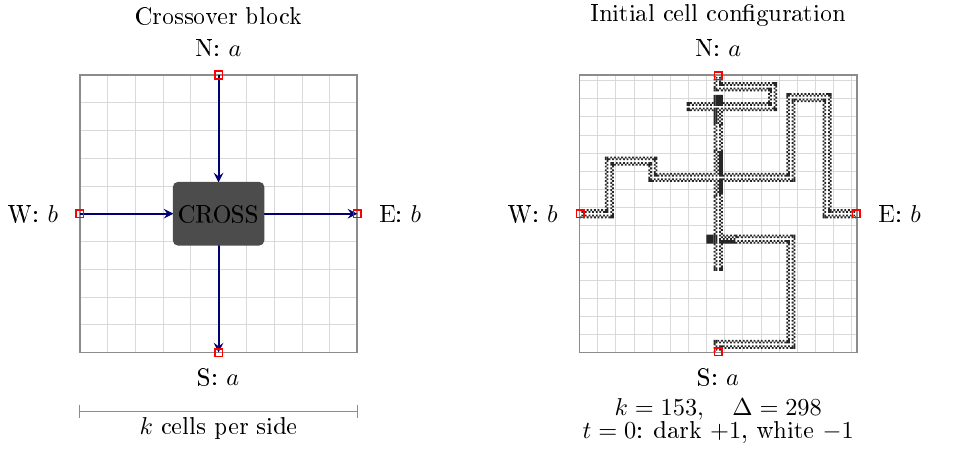}
\caption{Uniform crossover block with $(E,S)=(b,a)$. The north input continues south and the west input continues east. Left: logical assembly. Right: the initial spins of the $153\times153$ realization; the input-to-output delay is $\Delta=298$ on both outputs. External sources and sink caps lie outside the square. Port markers are enlarged; the physical grid shows every tenth line.}
\label{fig:block-crossover}
\end{figure}

%% file: gadgets/block-copy.tex
\begin{figure}[tbp]
\centering
\includegraphics[width=\linewidth]{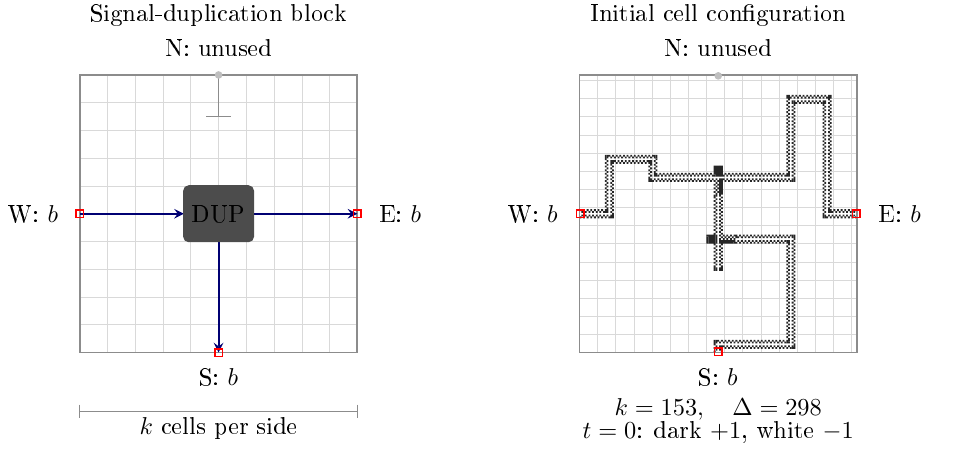}
\caption{Uniform signal-duplication block with $(E,S)=(b,b)$. The west input is used and the north side is closed. Reflection in the north-west/south-east diagonal gives the north-input variant. Left: logical assembly. Right: the initial spins of the $153\times153$ realization; the input-to-output delay is $\Delta=298$ on both outputs. External sources and sink caps lie outside the square. Port markers are enlarged; the physical grid shows every tenth line.}
\label{fig:block-copy}
\end{figure}

%% file: gadgets/fuse-idle.tex
\begin{figure}[!htbp]
\centering
\includegraphics[width=\linewidth]{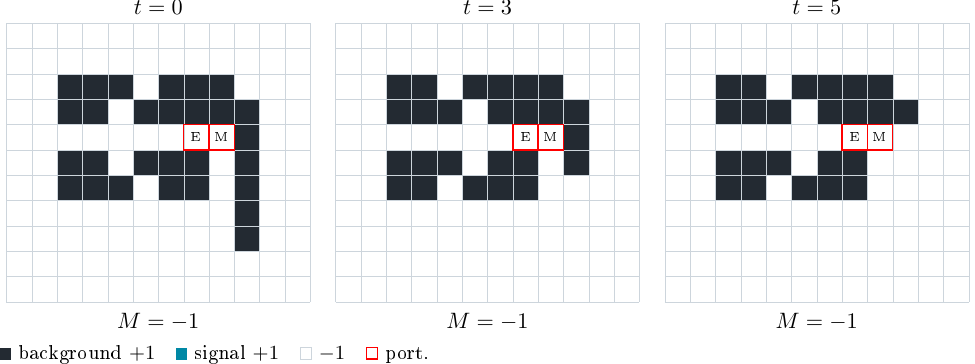}
\caption{Temporal fuse without an input source. The compact $8\times7$ configuration is initially armed, with terminal $E=(-1,0)$ and memory $M=(0,0)$. The autonomous clock erodes even without a signal: it can still accept an early input at $t=3$, but has disabled capture at $t=5$. This is a finite waiting window.}
\label{fig:fuse-idle}
\end{figure}

%% file: gadgets/fuse-dynamics.tex
\begin{figure}[!htbp]
\centering
\includegraphics[width=\linewidth]{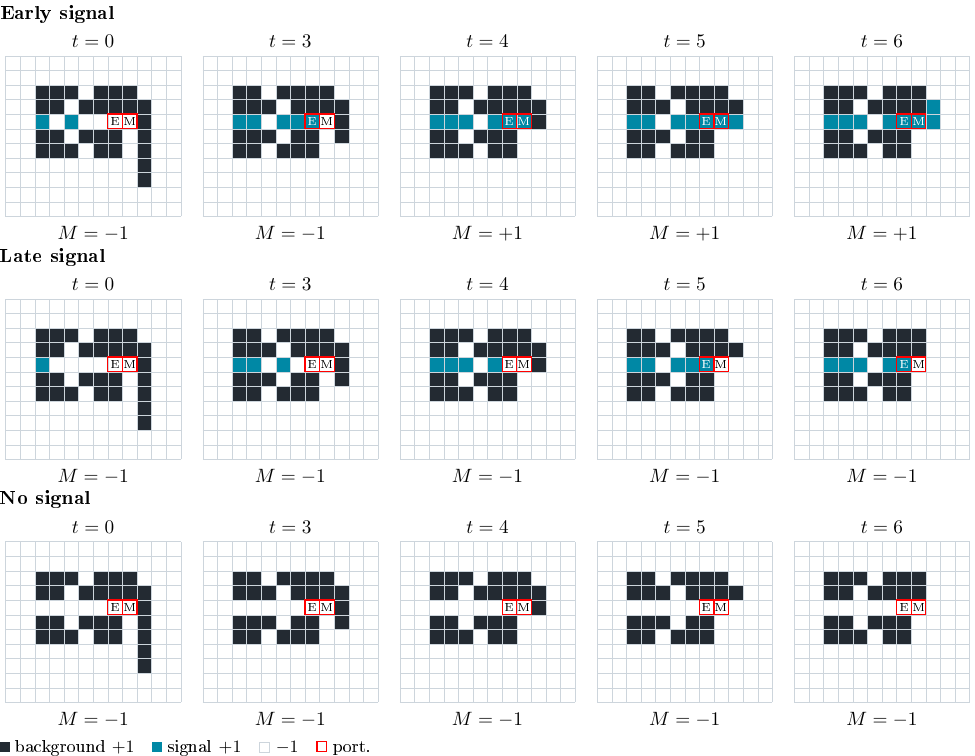}
\caption{Temporal fuse under early, late and absent input signals (top to bottom). An early arrival at $E$ at $t=3$ sets $M=+1$ permanently from $t=4$. A late arrival at $t=5$, or no arrival, leaves $M=-1$ permanently. The columns show initialization, early arrival, memory response, late arrival and the next step.}
\label{fig:fuse-dynamics}
\end{figure}